\documentclass[10pt,authoryear,review]{elsarticle}

\makeatletter
\def\ps@pprintTitle{%
 \let\@oddhead\@empty
 \let\@evenhead\@empty
 \def\@oddfoot{}%
 \let\@evenfoot\@oddfoot}
\makeatother

\usepackage{graphicx}
\usepackage{amssymb}
\usepackage{amsmath}
\usepackage{amsthm}
\usepackage{color}
\usepackage{dsfont}
\usepackage{calc}

\newsavebox\CBox
\newcommand\hcancel[2][0.5pt]{%
  \ifmmode\sbox\CBox{$#2$}\else\sbox\CBox{#2}\fi%
  \makebox[0pt][l]{\usebox\CBox}%
  \rule[0.5\ht\CBox-#1/2]{\wd\CBox}{#1}}

\usepackage{soul}
\usepackage{algorithm}% http://ctan.org/pkg/algorithm
\usepackage{algpseudocode}% http://ctan.org/pkg/algorithmicx

\usepackage{enumitem}
\setlist{leftmargin=5.5mm}

\usepackage{hyperref}

\newcommand{\mfN}{\mathfrak{N}}

\newcommand{\Q}{{\mathbb Q}}

\newcommand{\e}{{\varepsilon}}

\newcommand{\N}{{\mathbb N}}
\newcommand{\E}{{\mathbb E}}
\newcommand{\V}{{\mathbb V}}

\newcommand{\Fbb}{{\mathbb F}}

\newcommand{\calA}{{\mathcal A }}

\newcommand{\calF}{{\mathcal F}}
\newcommand{\calS}{{\mathcal S}}

\newcommand{\supp}{\text{supp}}

\newcommand{\dsR}{\mathds{R}}
\newcommand{\R}{\dsR}

\newtheorem{theorem}{Theorem}
\newtheorem{lemma}{Lemma}
\newtheorem{proposition}{Proposition}

\newtheorem{assumption}{Assumption}

\begin{document}

\begin{frontmatter}

\title{\textbf{Mixing LSMC and PDE Methods to Price Bermudan Options}\tnoteref{t1}
\\[1em]
\textit{SIAM J. Financial Mathematics, Forthcoming}
}
\tnotetext[t1]{SJ (RGPIN-2018-05705 and RGPAS-2018-522715) and KJ (RGPIN-2016-05637) would like to thank NSERC for partially funding this work. The authors would like to thank three anonymous referees for helpful comments that ultimately improved the paper. The first version of this paper was posted on SSRN on November 16, 2016, and is available at https://ssrn.com/abstract=2870962. This version: \today}

\author[author1]{David Farahany}
\ead{dfarahany@gmail.com}

\author[author2]{Kenneth Jackson}
\ead{krj@cs.utoronto.ca}
\address[author2] {Department of Computer Science, University of Toronto}

\author[author1]{Sebastian Jaimungal}
\ead{sebastian.jaimungal@utoronto.ca}
\address[author1] {Department of Statistical Sciences, University of Toronto}

\begin{abstract}
We develop a mixed least squares Monte Carlo-partial differential equation (LSMC-PDE) method for pricing Bermudan style options on assets under stochastic volatility. The algorithm is formulated for an arbitrary number of assets and volatility processes and we prove the algorithm converges almost surely for a class of models. We also introduce a multi-level Monte-Carlo/multi-grid method to improve the algorithm's computational complexity. Our numerical examples focus on the single ($2d$) and multi-dimensional ($4d$) Heston models and we compare our hybrid algorithm with classical LSMC approaches. In each case, we find that the hybrid algorithm outperforms standard LSMC in terms of estimating prices and optimal exercise boundaries.
\end{abstract}

\end{frontmatter}

\section{Introduction}

In recent years, mixed Monte Carlo-partial differential equation (MC-PDE) methods for European options have seen an increase in research activity. In the context of stochastic volatility (SV) models with one-way coupling, these methods revolve around simulating the SV process, computing an expectation by solving a PDE conditional on the volatility path, and averaging over paths. The approach has been around for some time as in \cite{HW} and \cite{Lewis} but has seen renewed interest in \cite{Ang:2013zl}, \cite{LLP}, \cite{LP}, \cite{DJM}, \cite{DJS} and \cite{cozmareis}.

In the context of high-dimensional European option pricing problems, under stochastic volatility, finite difference methods cannot be readily applied, and the correlations between the underlying processes often make the system non-affine which rules out Fourier-based quadrature methods. Also, full Monte-Carlo (MC) methods applied to such systems suffer from high variance and computational costs. An alternative is to find a middle ground between the two approaches where one simulates the underlying volatility processes and solves the resulting lower dimensional conditional PDEs, which may often be handled efficiently. This mixed method results in dimension reduction from the PDE perspective and variance reduction from the MC perspective.

As previous research which utilizes this strategy is focused on pricing European style options and addresses the dimension and variance reduction in computing relevant expected values, we analyse the mixed MC-PDE framework for Bermudan style options. The Bermudan context requires dealing with a high dimensional PDE between exercise dates, along with a high dimensional grid for accurately locating the exercise region; the latter being an issue that arises when moving from European to Bermudan options. When pricing Bermudan options the primary object of interest is the optimal stopping policy and the exercise boundaries that it defines. We note that one can always approximate the price of an American style option by considering a Bermudan option with high number of exercise dates as discussed in \cite{Bouchard}.

To deal with the above issues, we develop a hybrid method which mixes the least squares Monte Carlo (LSMC) approaches of \cite{LS} and \cite{TV} with PDE techniques. The essence of our version of a mixed LSMC-PDE algorithm is to
\begin{center}
\begin{minipage}{0.75\textwidth}
\begin{enumerate}
\setlength\itemsep{-0.25em}
\item simulate paths of the underlying SV process,
\item solve the conditional expectation (using a PDE approach) along each path,
\item regress these conditional expectations onto a family of basis functions over the volatility state-space.
\end{enumerate}
\end{minipage}
\end{center}
The algorithm may be viewed as an extension of \cite{TV} and reduces the monitoring of a high dimensional grid by replacing the volatility dimensions with a few regression coefficients. The approach provides variance reduction from the Monte-Carlo perspective, dimension reduction from the PDE perspective and alters the regression problems that are solved at each time step such that they are simpler than in standard LSMC. Our approach has its roots in \cite{LLP} where it is very briefly mentioned, but not analysed. Our contribution is a precise development of the algorithm, proof of convergence, discussion of complexity and complexity reduction methods, followed by a series of numerical examples.

When applied to SV problems, LSMC tends to be fairly inaccurate in determining the optimal exercise boundaries. In the literature, there have been a few direct modifications to LSMC, applied to SV problems, such as in \cite{Ludkovski2}, and \cite{ludkov} that address this issue. There have also been other types of probabilistic approaches such as in \cite{Ait} and \cite{Sircar}. The approach of \cite{Ait} is specific to the Heston model and appears to be non-applicable to models outside the affine class. \cite{Sircar} develops a highly efficient method for multi-scale SV models. The work of \cite{Ludkovski2}, \cite{ludkov} cast LSMC as a classification problem using various experimental designs and regression methods, and, as one of their examples, consider a one dimensional mean-reverting SV model. The approach of \cite{SGBM} also looks promising, although it typically requires a choice of basis functions for which one can compute (or approximate) expectations in closed form, and need to be developed case-by-case. It's also worth noting the work of \cite{particle}, which deals with the related problem of Bermudan option pricing under unobservable SV.

The remainder of this paper is organized as follows. In Section \ref{section:LSMC-PDEalg}, we set up our model and provide the basic mechanics of the algorithm. In Section \ref{section:LSMC-PDEalg} and \ref{section:algpf} we develop theoretical aspects of the algorithm such as formalizing the underlying probability spaces, deriving expressions for the regression coefficients, and showing the algorithm converges almost surely for pure-SV models. Section \ref{section:complexity} shows how Multi-Level Monte Carlo/multi-grids may be incorporated and also gives an overview of the overall algorithm's complexity. In Section \ref{section:NumericalExamples} we apply the algorithm to the Heston and multi-dimensional Heston model and present estimates of prices and optimal exercise boundaries. These results are also compared to finite difference and standard LSMC approaches.

\section{A Hybrid LSMC/PDE algorithm}
\label{section:LSMC-PDEalg}

\subsection{The Model}

We suppose the existence of a probability space $(\Omega,\Fbb,\Q)$ which may accommodate a $d_S + d_v$ dimensional stochastic process $(S_t,v_t) = (S^{(1)}_{t},...,S^{(d_s)}_{t},v^{(1)}_{t},...v^{(d_v)}_{t})_{t\in[0,T]}$ satisfying a system of SDEs with a strong, unique solution. We begin by defining mappings
\begin{alignat*}{5}
\mu_S   :  [0,T] \times   \R^{d_v} && \to \R^{d_S} \hspace{2cm} && \sigma_S  : [0,T] \times \R^{d_v} & \to \R^{d_S} \\
\mu_v   :  [0,T] \times  \R^{d_v}&&         \to \R^{d_v} \hspace{2cm} && \sigma_v  : [0,T] \times \R^{d_v} & \to \R^{d_v} \
\end{alignat*}
and a $d_S + d_v$-dimensional Brownian motion, $W=(W_t^S,W^v_t)_{t\in[0,T]}$, with correlation matrix, $\rho$. The process $(S_t,v_t)_{t\in[0,T]}$ is assumed to satisfy the following system of SDEs
\begin{alignat*}{6}
&dS_t = S_t \odot && \mu_S(t,v_t) && \ dt  \ +  S_t \odot  && \sigma_S(t,v_t)  \odot \  && dW^{S}_t  \\
&dv_t =           && \mu_v(t,v_t) && \ dt  \ +             && \sigma_v(t,v_t)   \odot \ && dW^{v}_t   \
\end{alignat*}
where for $x,y \in \R^m$ and  $x \odot y := \left(x_1 \  y_1, \ldots, x_m \ y_m\right)$ is element-wise product. As we shall see, our approach allows for an arbitrary (as long as it is positive semi-definite) stock-volatility correlation structure.

The above SDE classifies $(S_t,v_t)$ as a pure SV model, examples of which have been developed in \cite{Stein}, \cite{Heston},  \cite{fastMR}, and \cite{4o2Heston} among others. The work in this paper may also be extended easily to multi-factor pure SV models such as in \cite{dHeston}. While the algorithm we discuss in this paper may conceivably be applied to SV models with a non-linear local volatility (LV) component, our derivations, numerical examples, and proof of convergence pertain to only pure SV models.  Another property of $(S_t,v_t)$ is that it exhibits `one-way coupling': the SDE for $v_t$ may be simulated independently of $S_t$. We shall make this notion more precise in Section \ref{sec:derivcondpde}.

\subsection{Bermudan Option Pricing}

Let $\{t_0,...,t_M\} \subset [0,T]$ be an ordered set of exercise dates with $\Delta t_k = t_{k+1} - t_{k}$ and $h_{t_i}(S) : \R^{d_S} \to \R$ be our exercise function at each date. We often suppress the subscript $k$ in $\Delta t_k$ and $h_{t_k}$ to simplify notation when the context is clear. We also suppose the risk free rate is a constant, $r > 0$.

Valuing a Bermudan option requires developing an algorithm for evaluating
$V_t = \sup_{\tau \in T_t} \E^{\Q}\left[ \left.e^{-r  \tau}h_{\tau}(S_{\tau} )\right. \mid \calF_t \right]$ where $\calF_t = \calF^{S}_t  \vee \calF^{v}_t $ and $T_t$ is the set of $\mathcal F$-stopping times taking values in $\{t_k: k\in\{1,\dots,M\} \text{ and } t_k > t\}$. By the Markov property, $V_t$ depends only on $(t,S_t,v_t)$, and we can write $V_t=V(t,S_t,v_t)$ for some function $V:\mathds R_+\times \mathds R^{d_S}\times \R^{d_v}\mapsto \mathds R$.

At time $T = t_M$ we have $V(t_M,S_{t_M},v_{t_M}) = h_{t_M}(S_{t_M})$. We then define a new function, $C(t,S,v)$, denoted as the continuation value, at times $t_k$ by
\begin{equation*}
C(t_k,S,v) = e^{-r \Delta t_k} \;\E^{\Q}\left[ \left.V(t_{k+1},S_{t_{k+1}},v_{t_{k+1}}  )\right. \mid \left.S_{t_k} = S, v_{t_k} = v\right. \right].
\end{equation*}
By the discrete dynamic programming principle (DPP), for $k < M$ , we may express $V(t_k,S_{t_k},v_{t_k})$ as the maximum of the continuation value and the immediate exercise value at $t_k$:
\begin{equation*}
V_{t_{k}} = \max\left( h_{t_k}(S_{t_k}) \;, \;C(t_k,S_{t_k},v_{t_k}) \right)\;.
\end{equation*}
From this point on, for notational simplicity, we condense notation and replace our $t_k$ subscripts with simply $k$. We also replace arguments depending on time by subscripts.

\subsection{Algorithm Overview}

We now describe a hybrid-method for computing $V_k(S,v)$ which is based on the \cite{TV} approach, but uses conditional PDEs to incorporate dimensional and variance reduction. We begin by giving an intuitive explanation and provide a formal, pseudo-code based, description in \ref{app:algstate}.

We simulate $N$ paths of $v$ starting from the initial value $v_0$. Each path of $v_t$ over $[t_k,t_{k+1}]$ is represented as $[v]_{k}^{k+1}$. Given a product set $\mathcal{S} \subset \R^{d_S}$, we compute $V_k$ over the domain $\mathcal{S} \times \R^{d_v}$. The set $\calS$ is the domain of the conditional expectations that we compute; in practice, it is the grid for our numerical PDE solver. We suppose the discretized form of $\calS$ has $N_{s}$ points in each dimension so that there are $N^{d_S}_s$ points in total. Given the value $V_{k+1}(S_{k+1},v_{k+1})$ of the option at time $t_{k+1}$ we proceed to compute the continuation value at $t_k$. The algorithm begins at time $t_M = T$ where $V_M(S_M,v_M) = h_M(S_M)$.

\subsubsection{Solving along $S$ to obtain the pre-surface}\label{sec:step1}
For each simulation path $j\in\mfN$ ($\mfN=\{1,\dots,N\}$), we compute (beginning with $k=M-1$)
\begin{equation}
C_k^j(s_i):=e^{-r\Delta t}\;\E^{\Q} \left[ \; \left.V_{k+1}(S_{k+1},v_{k+1}) \; \right.| \;S_{k} = s_i\,, \;[v^j]_{k}^{k+1}  \; \right]
\label{eqn:condexp}
\end{equation}
for all $s_i \in \mathcal{S}$. These may be computed simultaneously over $\mathcal{S}$ for each path using a numerical PDE solver. The details of this step may be found in Section \ref{sec:derivcondpde}, below.

\begin{figure}
\centering
  \centering
  \includegraphics[width=\textwidth]{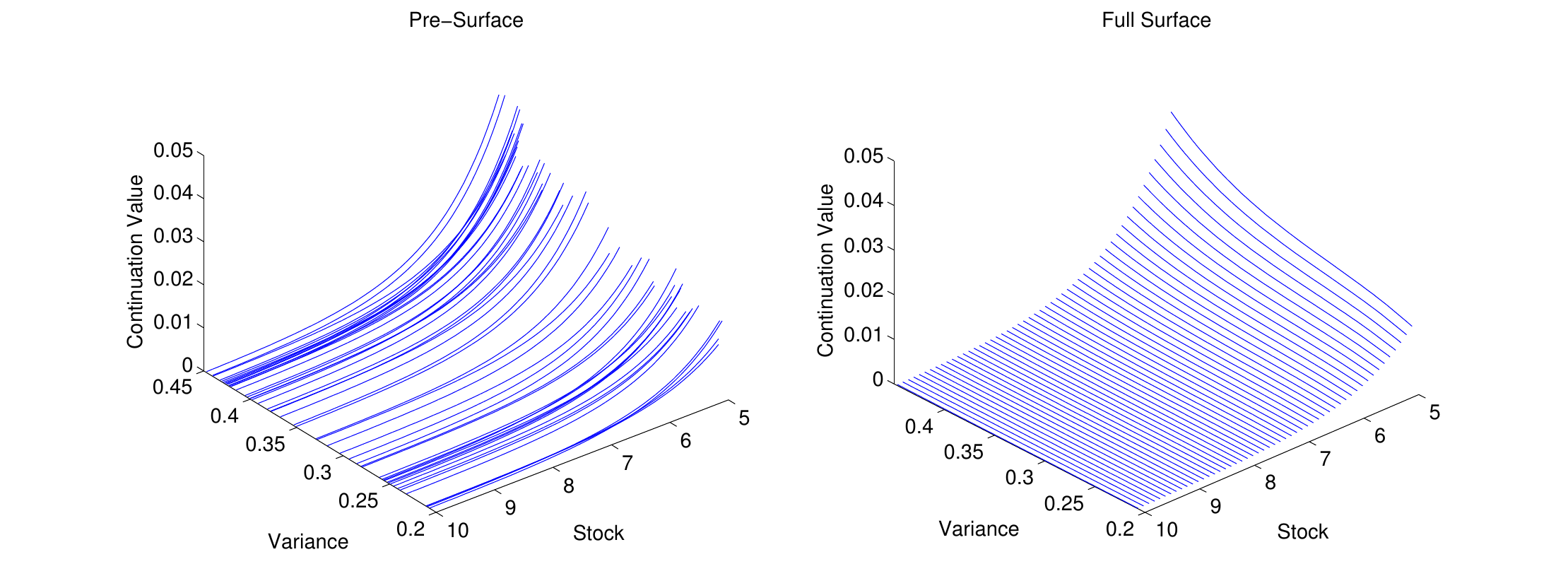}
  \caption{The pre and completed continuation surface. The pre-surface is generated by solving a PDE along each variance path. It is smooth along the $S$ axis, and noisy across the $v$-axis. The completed surface is generated by regressing across the $v$-axis.}
  \label{fig:test1}
\end{figure}

\subsubsection{Regress across $v$ to obtain the completed surface}\label{sec:step2}

For each $s_j \in \mathcal{S}$, from the previous step, we have $N$ realizations of the continuation value along each volatility path, i.e., $\{C_k^j(s_i)\}_{j\in\mfN}$.

Next, apply least-squares regression to project this onto a family  $\{\phi_m(\cdot)\}_{m=1}^{d_B}$ of linearly independent basis functions over our volatility space. This results in a vector of coefficients $a(s_i)$ of length $d_B$, and provides the continuation value at $S_{t_{k}} = s_i$ for any point in the volatility space as follows:
\begin{equation*}
C_{k}(s_i,v) = \sum_{m=1}^{d_B} a_{m,k}(s_i)\phi_m(v) = a_{k}(s_i)\cdot \phi(v).
\end{equation*}

\subsubsection{Obtaining the Option Price}

The price of the option is then given by $V_k(s_i,v) = \max( h_{k}(s_i), C_k(s_i,v) )$. These steps are repeated from \ref{sec:step1}  and \ref{sec:step2}  for all times $t_k$ where $k = M-1,...,1$.

%\subsection{Estimates on the Time Zero Price}

\subsubsection{A Direct Estimate on the Time Zero Price}

Since at time zero, there is only a single value for $v_0$, we obtain an estimate for our time-zero prices by
\begin{equation*}
V_{d,0}(s_i,v_0) = \max\left(h(s_i) \;,\; \tfrac{1}{N}\sum_{j=1}^{N}e^{-r\Delta t}\;\E^{\Q}\left[ \left. V_1(S_{1},v_1) \; \right.| \;  S_{0} = s_i\;,\; [v^j]_{0}^{1} \;   \right] \;\right)\,
\end{equation*}
which is often biased high. Following \cite{SGBM}, we call this the direct estimator.

\subsubsection{A Lower Estimate on the Time Zero Price}

Given our estimated regression coefficients, we obtain a sub-optimal exercise policy $\tau(t,S,v)$ defined on $\{t_1,...,t_{M-1} \} \times \calS \times \R^{d_v}$. Thus, we may define a lower estimate via the expectation
\begin{equation}
\E^{\Q} \left[ \; e^{-r \tau} h( S_{\tau} ) \mid  S_0, v_0 \; \right].
\label{eqn:lest}
\end{equation}
In traditional LSMC, one simulates a new independent set of paths $(S_t,v_t)$ to approximate (\ref{eqn:lest}). In the class of models we study, simulating both $S_t$ and $v_t$ undermines the variance reduction obtained by the algorithm and we instead use a hybrid approach.
%Thus, we take an approach that is typically used for European style options.

To this end, we denote the $t_k$ holding and exercise regions by $\Gamma_k$ and $\Gamma^c_k$, respectively. We then simulate $N$ new independent paths of $v_t$ on $[0,T]$, compute
\begin{equation}
\E^{\Q}\left[ \; e^{-r \tau } h( S_{\tau})  \mid  S_0, [v^j]_0^T   \; \right]
\label{eqn:lestcond}
\end{equation}
via a PDE approach for $j\in\mfN$, and take the average. To compute (\ref{eqn:lestcond}), for each $j\in \mfN$, first set $V^j_M(S,v) = h_M(S)$. Next, compute
\begin{equation*}
U^j_{M-1}(S) = e^{-r\Delta t}\,  \E^{\Q} \left[   V^j_M(S_M,v_M) \mid S_{t_{M-1}} = S , [v^j]_{M-1}^M  \right]
\end{equation*}
via a PDE method for all $S \in \calS$. The option price at time $t = t_{M-1}$ is then given by
\begin{equation*}
V^j_{M-1}(S) = U^{j}_{M-1}(S) \cdot I_{(S,v^j_{M-1}) \in \Gamma_{M-1}} +  h_{M-1}(S) \cdot I_{(S,v^j_{M-1}) \in \Gamma^{c}_{M-1}}
\end{equation*}
After repeating this procedure for times $k = M-2,...,1$ we obtain the lower estimate
\begin{equation}
V_{l,0}(S) = \frac{1}{N} \sum_{j=1}^{N} e^{-r \Delta t}\,  \E^{\Q} \left[ \;   V^j_1(S_1,v_1)  \mid S_{t_{1}} = S , [v^j]_{0}^1 \;   \right]
\label{eqn:lowestim}
\end{equation}
for all $S \in \calS$.

\subsection{Numerical Computation of Conditional Expectations}
\label{sec:derivcondpde}

In this section we provide the details on how to numerically compute expressions of the form
\begin{equation*}
\E^{\Q} \left[ \; \left. h(S_{t_{n+1}},v_{t_{n+1}}) \; \right.| \;S_{t_n} = S \,, \;[v]_{t_n}^{t_{n+1}}  \; \right]
\end{equation*}
using numerical solutions of PDEs. Our derivation of the conditional PDE follows the methodology of \cite{DJS}, and our numerical method to solve the conditional PDEs is Fourier Space Time-Stepping (FST) as developed in \cite{JJS}.

\subsubsection{Change of Variables}

We begin by converting our asset price SDE into log-coordinates $X_t:= \log S_t$ yielding the system
\begin{alignat*}{3}
dX_t =&& (\mu_S(t,v_t) - \frac{1}{2}\sigma_S(t,v_t)^2 ) \ && dt + \sigma_S(t,v_t) \odot dW^S_t \\
dv_t =&&     \mu_v(t,v_t) \ && dt + \sigma_v(t,v_t) \odot dW^v_t \
\end{alignat*}
where $\sigma_S(t,v)^2 := \sigma_S(t,v)\odot\sigma_S(t,v)$. Next, applying the Cholesky decomposition to the matrix $\rho = [\rho_{ij}]$, we find an upper triangular matrix $A = [a_{ij}]$ and Brownian motion $\widetilde{W} := (\widetilde{W}^S,\widetilde{W}^v)$ such that $W= A \cdot \widetilde{W}$ and $\widetilde{W}$ has independent components. Next, we decompose $A$ in block-form as
\begin{equation*}
A=
\left[
\begin{array}{c|c}
A_{SS} & A_{Sv} \\
\hline
A_{vS} & A_{vv}
\end{array}
\right]
\end{equation*}
where
\begin{itemize}
\setlength\itemsep{-0.25em}
\item $[A_{SS}]_{i,j} = A_{i,j}$  for $1\leq i \leq d_{S}$,\ $ 1\leq j \leq d_{S}$
\item $[A_{Sv}]_{i,j} = A_{i,j+d_S}$   for $1\leq i \leq d_{S}$,\ $1\leq j \leq  d_v$
\item $[A_{vS}]_{i,j} = A_{i+d_S,j}$  for $ 1\leq i \leq d_v$,\  $1\leq j \leq d_{S}$
\item $[A_{vv}]_{i,j} = A_{i +d_S, j+d_S}$  for $1 \leq i \leq  d_v$, \ $ 1\leq j \leq d_{v} $.
\end{itemize}

We then re-write the system as:
\begin{alignat*}{6}
dX_t =&& (\mu_S(t,v_t) - \tfrac{1}{2}\sigma_S(t,v_t)^2 ) \ dt && + \sum_{j=1}^{d_S} \sigma_S(t,v_t) \odot [A_{SS}]_j \odot d\widetilde{W}_t^{S,j} && + \sum_{j=1}^{d_v} \sigma_S(t,v_t) \odot  [A_{Sv}]_j \odot d\widetilde{W}_t^{v,j}  \\
dv_t = && \mu_v(t,v_t) \ dt && \ &&+  \sum_{j=1}^{d_v} \sigma_v(t,v_t) \odot [A_{vv}]_j \odot d\widetilde{W}_t^{v,j}    \
\end{alignat*}
where $[A_{\alpha\beta}]_i$ is the $i$th column of $A_{\alpha \beta}$.
% noting the associativity between $\odot$ and matrix multiplication.

In this form, $v_t$ may be simulated first, then inserted into $X_t$. To this end, define two new processes $Y_t$, $Z_t$ on $[t_n,t_{n+1}]$
\begin{align*}
Y_t &= X_{t_n} + \int_{t_n}^t  (\mu_S(t,v_s) - \frac{1}{2}\sigma_S(t,v_t)^2 ) \ ds + \sum_{j=1}^{d_S} \int_{t_n}^t  \sigma_S(s,v_s) \odot [A_{SS}]_j \odot d\widetilde{W}_s^{S,j}\\
Z_t &=  \sum_{j=1}^{d_v}  \int_{t_n}^t \sigma_S(t,v_s) \odot [A_{Sv}]_j \odot d\widetilde{W}_s^{v,j}  \
\end{align*}
and note that $X_{t_n} = Y_{t_n}$ and $X_t=Y_t+Z_t$. Next, define the function $g(X,v) := h(\exp(X),v)$ so that
\begin{align*}
\E^{\Q}\left[ h(S_{n+1},v_{n+1} ) \ \mid \ S_n = S, [v]_n^{n+1}  \right] &= \E^{\Q}\left[  g(X_{n+1},v_{n+1}) \ \mid \ X_n = x, [v]_n^{n+1}  \right] \\
&= \E^{\Q}\left[  g(Y_{n+1} + Z_{n+1},v_{n+1}) \ \mid \ Y_{n} = y, [v]_n^{n+1}   \right].  \
\end{align*}

\subsubsection{Conditional PDE and Fourier Solution}

Finally, we re-write our original expectation in terms of the function $u(t_n,y)$
\begin{equation*}
u(t_n,y) = \E^{\Q}\left[ \ g(Y_{t_{n+1}} + Z_{t_{n+1} } ,v_{t_{n+1}}) \ \mid \ Y_{t_n} = y , \ [v]_{t_n}^{t_{n+1}}  \ \right]
\end{equation*}
treating $[v]_{t_n}^{t_{n+1}}$ as a deterministic path for $y \in \log \calS$, $t \in [t_n,t_{n+1}]$. By the Feynman-Kac theorem, the function $u(t_n,y)$ may be written as the solution of the following PDE
\begin{equation}
\label{eq:mainpde}
\left\{
%\begin{array}{rl}
\begin{split}
\partial_t u(t,y)+  \sum_{i=1}^{d_S} \ a_i(t) \partial_{y_i} u(t,y)
%\qquad & \\
+\tfrac{1}{2}  \sum_{i=1}^{d_S} \sum_{j=1}^{d_S} \ b_{i,j}(t) \  \partial_{y_i,y_j}^2 u(t,y)
&= 0\,,
\\[0.25em]
u(t_{n+1},y)  &= g(y + Z_{t_{n+1}},v_{t_{n+1} })\,
\end{split}
%\end{array}
\right.
\end{equation}
where $a(t) = \mu_S(t,v_t) - \frac{1}{2}\sigma_S(t,v_t)^2$ and $b_{i,j}(t) = \sum_{j=i}^{d_S} [\sigma_S(t,v_t)]_{j}^2  [A_{SS}]^2_{i,j} $.

Applying the Fourier transform to (\ref{eq:mainpde}) over $\R^{d_S}$ we obtain
\begin{equation}
\left\{
%\begin{array}{rl}
\begin{split}
\partial_t \widehat{u}(t,\omega) +  \left( i \sum_{j=1}^{d_S} \ a_j(t)  \omega_j -  \tfrac{1}{2}\sum_{i=1}^{d_S} \sum_{j=1}^{d_S} \ b_{i,j}(t) \omega_i\omega_j  \right) \widehat{u}(t,y)
%\\
&= 0\,,  \\[0.25em]
\widehat{u}(t_{n+1},\omega) &= e^{i \omega \cdot Z_{t_{n+1} } } \widehat{g}(\omega,v_{t_{n+1}})\,,
\end{split}
%\end{array}
\right.
\end{equation}
where $\omega \cdot Z_{t_{n+1} }$ denotes the inner product on $\R^{d_S}$. For each $\omega$ in Fourier space, this ODE may be solved analytically to obtain the following
\begin{align*}
\widehat{u}(t,\omega) &= \widehat{u}(t_{n+1},\omega) \cdot \exp\left(  i\sum_{i=1}^{d_S} \omega_i\int_{t}^{t_{n+1}} a_i(s) ds  - \frac{1}{2} \sum_{i=1}^{d_S}\sum_{j=1}^{d_S} \omega_i \omega_j \int_t^{t_{n+1}} b_{i,j}(s) ds   \right)\\
&= \widehat{g}(\omega,v_{t_{n+1} } ) \cdot \exp\left( i \omega \cdot Z_{t_{n+1}} + i\sum_{i=1}^{d_S} \omega_i\int_{t}^{t_{n+1}} a_i(s) ds  - \frac{1}{2} \sum_{i=1}^{d_S}\sum_{j=1}^{d_S} \omega_i \omega_j \int_t^{t_{n+1}} b_{i,j}(s) ds   \right). \
\end{align*}
\subsubsection{Numerical Approximation}
Defining the characteristic exponent
\begin{equation*}
\Psi(\omega,t_{n},t_{n+1} )  = i \omega \cdot Z_{t_{n+1}} + i\sum_{i=1}^{d_S} \omega_i\int_{t_n}^{t_{n+1}} a_i(s) ds  - \frac{1}{2} \sum_{i=1}^{d_S}\sum_{j=1}^{d_S} \omega_i \omega_j \int_{t_n}^{t_{n+1}} b_{i,j}(s) ds
\end{equation*}
or more compactly $\Psi_{n,n+1} := \Psi(\omega,t_{n},t_{n+1} )$ and using the FST's discretization methodology with fast Fourier transforms (FFT), we have the following recursion
\begin{equation*}
u_n = \text{FFT}^{-1}_{d_S}\left[  \text{FFT}_{d_S}\left[ g_{n+1}  \right] \exp\left( \Psi_{n,n+1}  \right)    \right]
\end{equation*}
where $g_n, u_n$ are the discretizations of $g(t_n, y), u(t_n,y)$, respectively, and $\text{FFT}_{d_S}$ denotes the $d_S$-dimensional Fast-Fourier Transform. A change of variables back to $S$-space then provides an numerical approximation of the original expectation.

\subsection{Discussion of Algorithm}

We refer the reader to \ref{app:algstate} for a pseudo-code based formal description.

%Although we solve a PDE over thousands of paths of $v_t$ over each time interval $[t_n,t_{n+1}]$ and solve a linear-regression problem for each $s_i \in \calS$ and $t_n$, the computational costs and run time are not as high as they may seem. First, the PDEs over each volatility path, and the regressions at $s_i \in \calS$, are independent and can be parallelized. Also, based on (\ref{eqn:estcoef}), the regression  problems at time $t_n$ require only one matrix inversion, and its result is applied to each of the $N_s^{d_S}$ regression sites. Finally, in Section \ref{section:pathred}, we discuss two  methods which allow us to reduce the algorithm's complexity.

%We immediately see that our algorithm provides dimensional reduction from the PDE and regression perspective, and variance reduction from the MC perspective for a fixed simulation budget. If one employs a fully numerical scheme to solve the conditional PDEs, the algorithm is, in principle, capable of handling $3 + n$ dimensional problems where one solves a PDE over three dimensional asset space and simulates $n$ volatility variables.

In Section \ref{section:NumericalExamples}, we show that the algorithm accurately, for a given computational budget, determines the time-zero value surface and optimal exercise regions. This may be attributed to the presence of our PDE grid, $\calS$. When the conditional PDEs are solved along each path, we obtain our pre-surface as described in Subsection \ref{sec:step1}. At this point one has two choices: global or local regression. Our regression approach can be viewed as a special type of dimension reduced, local regression which is tailored to the presence of $\calS$ and is equivalent to local regression onto $N_s^{d_S}$ carefully chosen regions. If $N_s = 128$ and $d_S = d_v = 2$, we are regressing onto $16,384$ families of basis functions at the cost of inverting a single matrix of size $d_B \times d_B$ where $d_B$ is about $10$, and a single matrix-multiplication for each $S \in \calS$; the latter step introducing non-trivial costs as $d_S$ increases. The fact that we must only invert a single regression matrix, independent of $S \in \calS$, at each time step will be made more clear in Section \ref{section:algpf}. Also, $C_n(s_j,v)$ is typically simple to fit as a function of $v$ and seems to resist the Basis Selection Problem.

Working with $\calS$ has other advantages as well. In comparison to standard approaches to LSMC, there is a fundamental shift in how we compare the continuation value to the exercise value and locate the exercise boundary. At time $n$, when setting the value of $V_n$ for each $s_i \in \calS$ we have
\begin{equation*}
V^N_n(s_i,v) = \max \left( h_n(s_i), C^{N}_n\left(s_i,v\right) \right)\,,
\label{eqn:compare}
\end{equation*}
and note that $h_n(s_i)$ is a deterministic constant as opposed to a function of a random variable. Thus, we have reduced the problem of locating the boundary from a global problem over $(S,v)$-space to a sequence of lower dimensional problems which are simpler in nature and exhibit less noise. Also, compared to LSMC, instead of an estimate of $V_n$ at a single point, $(S_0,v_0) \in \R^{d_S + d_v}$, one obtains a semi-global solution, i.e. a value for all $S \in \calS$ and $v_0 \in \R^{d_v}$, owing to the PDE aspect of the algorithm. Obtaining a solution for all $S \in \calS$ allows us to compute sensitivities with respect to $S$ with little extra computation.

Generally speaking, deriving the conditional PDE is non-trivial and to the best of our knowledge there are two approaches in the literature: the drift discretization method in \cite{LLP}, \cite{cozmareis}, \cite{DJM} and the conditionally-affine decomposition of \cite{DJS}. The drift-discretization approach is fairly general and may be applied to a wide class of models, including those with a non-linear LV component, i.e. stochastic LV models. The resulting PDEs may then be solved using a finite difference type approach as in \cite{LLP}. We follow the conditionally-affine approach as it applies to our pure-SV model and avoids time stepping error and drift discretization.

Finally, from the PDE perspective, we begin with a deterministic PDE problem and, using one-way coupling, we convert it to a lower-dimensional stochastic PDE (SPDE) problem. The SPDEs we solve are essentially Black-Scholes equations where the coefficients and terminal condition depend on the random simulated process, $v$. The multi-period aspect of the problem is then handled using regression as outlined above. In light of the connection to SPDEs, in Section \ref{section:complexity}, we develop a generic Multi-Level Monte-Carlo (MLMC) scheme which completes the algorithm.

\section{Theoretical Aspects of the Algorithm}
\label{section:algpf}

In this section we provide a more theoretical perspective on the mechanics of the algorithm. We begin in Subsection \ref{subsection:condexptnsamplingspace}, where we introduce notions that are needed for the convergence result. In Section \ref{subsection:condexptnsamplingspace} we derive expressions for the estimated regression coefficients. Subsection \ref{subsection:truncationscheme} introduces a truncation scheme to ensure the algorithm is well-defined and Subsection \ref{subsection:asconvergence} shows the algorithm convergence almost surely under certain extra continuity conditions.

\subsection{Notation}

We denote Euclidean norms of elements $x \in \R^n$ or $ x \in \mathds{M}_{n\times m}(\R)$ by $|x|$. Given any function $h: \R^n \to \R$, we write $||h||_{\infty} := \sup_{x \in \R^n} |h(x)|$ and $\supp \ h  :=  \text{cl}\{ x \in \R^n \ | \   |h(x)| > 0  \}  $. Letting $X$ be an open subset of $\R^{n}$ we define $C_0(X)$ to be the set of continuous functions on $f: X \to \R$ that vanish at infinity. By vanish at infinity, we mean that for every $\e > 0$, the set $\{ x \in X \ | \ |f(x)| > \e\}$ is compact. We also let $C_c(X)$ be the set of compactly supported, continuous functions on $X$.

\subsection{Conditional Expectations and the Sampling Space}
\label{subsection:condexptnsamplingspace}
\subsubsection{Filtrations and Conditional Probabilities }

We recall the existence of a probability space $(\Omega,\Fbb,\Q)$ which accommodates the process $(S_t,v_t)$ satisfying a system of SDEs with a strong, unique solution. In this section, we make rigorous the notion of conditioning on a volatility path.

Let $\calF^{v}_{s,t}=\sigma( v_u)_{u\in[s,t]}$, $\calF^S_{s,t}=\sigma( S_u)_{u\in[s,t]}$ and $\calF^{W^v}_{s,t}=\sigma(W^v_u)_{u \in[s,t]}$ i.e., the natural filtrations generated by $v,S$ and $W^v$, respectively. To extend this notation, we sometimes write $\calF^{Z}_t := \calF^{Z}_{0,t}$ for some process $Z$. Given $ t_n \in [0,T]$, we define a new class of (conditional) probability measures $\Q_{t_n,S}$ via $\Q_{t_n,S}(B) = \Q( B \ | \ S_{t_n} = S)$ for $B \in \calF^S_{t_n,T} \vee \calF^v_{t_n,T}$.

\subsubsection{Conditioning on $[v]$}
In this Section, we make concise the notion of conditioning on a volatility path $[v]$.

For a realization of $v_t$ and $W^v_t$ on $[t_n,t_{n+1}]$, which we denote as $[v]_{t_n}^{t_{n+1}}$, there exists a finite-dimensional statistic of the path,
\begin{equation*}
\Lambda_n : C_0([t_n,t_{n+1}])^{d_v + d_{W^v}} \rightarrow \R^{d_{\Lambda}},
\end{equation*}
such that the following Markovian-like relation holds
\begin{align}
&\E^{\Q}[ \ h(S_{t_{n+1}}, v_{t_{n+1}}  ) \ | \ S_{t_n} = S \ , \  \calF^{W^v}_{t_{n} ,t_{n+1} } \   ]  \nonumber \\
&\quad= \E^{\Q}[ \  h(S_{t_{n+1}}, v_{t_{n+1}}   ) \ | \ S_{t_n} = S \ , \ \Lambda_n([v]^{t_{n+1}}_{t_n}) , \ v_{t_{n+1}} \ ]. \
\label{eqn:markov}
\end{align}
This can be seen from the solution to the derived conditional PDE in Section \ref{sec:derivcondpde} which depend only on the vectors $\left[\int_{t_n}^{t_{n+1}} a_i(s) \ ds\right]_i$, $\left[\int_{t_{n}}^{t_{n+1}} b_{i,j}(s) \ ds \right]_{i,j} $, and $Z_{t_{n+1}}$. In our algorithm, when computing regression coefficients, we encounter expectations of the form
\begin{equation*}
\E\left[ \, \left.\phi(v_{t_n} ) h(S_{t_{n+1}},v_{t_{n+1}}) \, \right| \ S_{t_n} = S \ , \ \calF^{W^v}_{t_{n},t_{n+1} } \, \right]
\end{equation*}
where $\phi: \R^{d_v} \to \R$ which depend on the statistic $\Theta_n([v]_n^{n+1}) := (v_{t_n}, \Lambda_n([v]_n^{n+1}), v_{t_{n+1}}). $ To condense notation, we again replace $\Theta_n([v]_n^{n+1})$ with $[v]_{t_{n}}^{t_{n+1}}$ and simply write
\begin{equation*}
  \E^{\Q}\left[ \left. \  \phi(v_{t_n})  h( \ S_{t_{n+1}},   v_{t_{n+1}} \  ) \ \right| \ S_{t_n} = S \ , \  [v]^{{t_{n+1} }}_{t_n}   \ \right]:=
  \E^{\Q}\left[ \  \left. \phi(v_{t_n})  h(S_{t_{n+1}}, v_{t_{n+1}} ) \ \right| \ S_{t_n} = S \ , \Theta_n([v]_n^{n+1})  \ \right].
\end{equation*}

\subsubsection{Re-expressing the Conditional Expectation }

Equation (\ref{eqn:markov}) gives rise to the mappings $G_{f,t_n,S}$, defined by
\begin{equation}
\begin{split}
G_{f,t_n,S} & : \R^{d_{\Theta}} \to \R \,, \\
G_{f,t_n,S}( \theta ) & = \E^{\Q} \left.[ \ f(S_{t_{n+1} }, v_{t_{n+1}},v_{n} ) \ | \ S_{t_n} = S, \  \Theta_n = \theta \ \right.] \
\end{split}
\label{eqn:Gfn}
\end{equation}
and the conditional probability measures $\Q_{t_n,S,\theta}$ defined via
\begin{equation*}
\Q_{t_n,S,\theta}(B) : = \Q( \ B \mid S_{t_n} = S, \  \Theta_n = \theta  )  \
\end{equation*}
for $B \in \sigma( S_{t_{n+1}} ) \vee \sigma(v_{t_{n+1}}) \vee \sigma(v_{t_n})$. Letting $\widetilde{\Q}_{ \Theta_n}$ denote the distribution of $\Theta_n$ on $\R^{d_{\Theta}}$ our conditional expectation may be written as
\begin{equation*}
\E^{\Q}\left.[ f( S_{t_{n+1}}, v_{t_{n+1}}, v_{t_{n} }   ) \mid S_{t_n} = S, \Theta_n = \theta \right.] \  =   \int_{\Omega} f(S_{t_{n+1} }(\omega), v_{t_{n+1}}, v_{t_n}  ) \ d\Q_{t_n , S,\theta} (\omega)
\end{equation*}
so that we have the following relation
\begin{align}
& \E^{\Q}\left.[ \ f(S_{t_{n+1}},  v_{t_{n+1}} , v_{t_n}) \mid S_{t_n} = S \right.]  \nonumber \\
&\quad= \int_{\R^{d_{\Theta}} } \int_{\Omega} f(S_{t_{n+1} }(\omega), v_{t_{n+1}},v_{t_n}  ) \ d\Q_{t_n , S,\theta} (\omega ) \ d\widetilde{\Q}_{\Theta_n} (\theta) . \
\label{eqn:mrelation}
\end{align}

\subsubsection{Inherited Sampling Probability Space}
\label{sec:SamplingSpace}

A consequence of one-way coupling is the ability to simulate paths of $v_t$ independently of $S_t$. We now make sense of the notion of an iid collection of sample paths of $v_t$. Since we only realize $v_t$ through the statistics $\Theta_n$, we only describe how to generate iid copies of $\Theta_n$.

Denote the ordered subset $\{0 = t_0,...,t_n,t_{n+1},...,t_N = T \} \subset [0,T] $, and let $\{ [t_n,t_{n+1}] \}_{n=0}^{N-1}$ be the corresponding intervals. Given a path $v_t$ on $[0,T]$, define the $d_{\Theta} \times N$ dimensional matrix
\begin{equation*}
\Theta([v]) = \left.[ \ \Theta_0( [v]^{1}_0) \ , \ldots , \ \Theta_{N-1}( [v]^{N}_{N-1})  \right.].
\end{equation*}
This random matrix induces a measure $\widetilde{\Q}_{\Theta}$ on $\R^{Nd_{\Theta}}$. Given $\widetilde{\Q}_{\Theta}$, we introduce a new probability space $( \Omega', \calF',\Q')$ equipped with a collection of independent random matrices
\begin{equation*}
\left.\{ \Theta([v^j] ) \right.\}_{j=1}^{\infty}\,,
\end{equation*}
such that each $\Theta([ v^j ]) $ has distribution $\widetilde{\Q}_{\Theta}$ on $\R^{Nd_{\Theta} }$. This construction follows from Kolmogorov's Extension Theorem applied to measures on $\R^{Nd_{\Theta} } $ (for a proof on $\R$ see \cite{Durrett}, for more general spaces see \cite{AB}). It then follows that each column $\Theta_n([v^j]_n^{n+1} )$ has distribution $\widetilde{\Q}_{\Theta_n}$ on $\R^{d_{\Theta} }$. Although the process $S_t$ is not defined on $\Omega'$, we can still compute relevant expectations involving this process using $G_{f,t_n,S}(\Theta)$ as defined in (\ref{eqn:Gfn}). We also recall that a random variable on $\Omega'$ corresponds to an equivalence class of mappings $X: \Omega' \to \R$ that agree $\Q'$-a.s.

\subsubsection{Computing Limits}

When taking limits in our algorithm, we consider expressions of the form
\begin{equation*}
\lim_{N\to\infty} \frac{1}{N}\sum_{j=1}^N \E^{\Q}\left.[ \ f(S_{n+1},v_{t_{n+1} },v_{t_n }) \ \mid \ S_{t_n} = S, [v^j]_n^{n+1}  \right],
\end{equation*}
where $f$ is bounded, which converge to $\E^{\Q'}\left.[    G_{f,t_n,S}\left.( \Theta\left.([v]_{t_n}^{t_{n+1}} \right.)   \right.)    \right.]$ a.s. under $\Q'$ by the Strong Law of Large Numbers (SLLN). For our purposes, however, we require convergence to
\begin{equation*}
\E^{\Q}\left[ \ f(S_{t_{n+1}},v_{t_{n+1}},v_{t_n} ) \ \mid \ S_{t_n} = S \right].
\end{equation*}
To establish the equivalence between these expressions:
\begin{align}
\E^{\Q'}\left.[    G_{f,t_n,S}\left.( \Theta\left.([v]_{t_n}^{t_{n+1}} \right.)   \right.)    \right.] & = \int_{\Omega'} G_{f,t_n,S}(\Theta_n([v(\omega')]_n^{n+1}  )   ) \  d\Q'(\omega') \nonumber\\
&= \int_{\R^{d_{\Theta}}} G_{f,t_n,S}(\theta) \ d\widetilde{\Q}_{\Theta_n}(\theta) \nonumber \\
& = \int_{\R^{d_{\Theta}}} \E^{\Q}\left.[ f( S_{t_{n+1}}, v_{t_{n+1}}, v_{t_n}   ) \mid S_{t_n} = S, \Theta_n = \theta \right.] \ d \widetilde{\Q}_{\Theta_n}(\theta) \nonumber   \\
& = \int_{\R^{d_{\Theta}}}   \int_{\Omega} f(S_{t_{n+1} }(\omega), v_{t_{n+1}}, v_{t_n}  ) \ d\Q_{t_n , S,\theta,n} (\omega)  \ d \widetilde{\Q}_{\Theta_n}(\theta) \nonumber  \\
& = \E^{\Q}\left.[ f(S_{t_{n+1}}, v_{t_{n+1}} , v_{t_n}  ) \mid S_{t_n} = S \right.] \label{eqn:measequiv} \
\end{align}
where the second equality follows from $\Theta([v]_{t_n}^{t_{n+1}}) $ being $ \widetilde{\Q}_{\Theta_n}$ distributed.

\subsection{Continuation Functions}
\label{subsection:continuationfunctions}

We remind the reader that for each $n$, the random variables $\{[v^j]_n^{n+1}\}_{j=1}^{\infty}$ are defined on the space $\Omega'$ and are iid (see the discussion in Section \ref{sec:SamplingSpace}). For notational convenience, we suppose the risk free rate is $0$. As our algorithm is based on the \cite{TV} approach to LSMC, many of our expressions are similar.

\subsubsection{Idealized Continuation Functions}

For each $n$, we consider a family of \textit{idealized continuation functions}, $C_{n}$, which are constructed by means of backwards induction. We begin by writing $C_M \equiv 0$ and $C_n(S,v) = a_n(S)\cdot \phi(v)$ for $n < M$ where $a_n(S)$ results from regressing the random variable
\begin{equation*}
 \E^{\Q}\left[ \left.\max(h_{n+1}(S_{n+1}),C_{n+1}(S_{n+1},v_{n+1}) )\right. \mid \left.S_{n} = S, \ v_n  \right.  \right] \ \ \text{onto the basis} \ \  \{\phi_m(\cdot)\}_{m=1}^{d_B}
\end{equation*}
for each $S \in \calS$.
The coefficient vector $a_n(S)$ is the vector that minimizes the mapping $H_{n,S}: \R^{d_B} \to \R$ defined by
\begin{equation}
H_{n,S}(a) = \E^{\Q}\left[ \left( \E^{\Q}\left[ \left.f_{n+1}(S_{n+1},v_{n+1})\right. \mid \left. S_n = S, v_n\right. \right]  - a \cdot \phi(v_n)  \right)^2 \mid  \left.S_n = S\right. \right]
\label{eqn:functionH}
\end{equation}
where
\begin{equation}
\begin{split}
f_{n}(S,v) & = \max( h_n(S), C_{n}(S,v) ) \ \text{ for } \ n < M\;, \text{ and }\,  \\
f_{M}(S,v) & = h_M(S)\;. \
\end{split}
\label{eqn:functionf}
\end{equation}
To minimize $H$, we obtain the first order conditions and obtain the normal equations, resulting in the coefficients
\begin{equation*}
a_n(S) = A_n^{-1}  \cdot   \E^{\Q}\left[  \left.\phi(v_n) f_{n+1}(S_{n+1},v_{n+1})\right. \mid \left. S_{n} = S\right.   \right]
\end{equation*}
where $A_n = \E^{\Q}[ \phi(v_{t_n}) \phi(v_{t_n})^{\intercal} ] $.

\subsubsection{Almost-Idealized Continuation Functions}

Next, for a fixed $t_n$, we define a new type of continuation value, called the \textit{almost-idealized continuation functions}, $\widetilde{C}_n^{N}(S,v) = \widetilde{a}_n^N(S) \cdot \phi_n(v)$. These random variable are obtained by running the dynamic programming algorithm with the idealized continuation value at all times $k = M,...,n+1$. At time step $n$ we then estimate $\widetilde{a}^{N}_n(S)$ using our $N$ paths of $v_t$ and future idealized continuation values. This gives us the following regression coefficients:
\begin{equation*}
\widetilde{a}^{N}_n(S) = \big[ A^{N}_n \big]^{-1}   \frac{1}{N} \sum_{j=1}^{N}   \ \phi(v^j_n) \cdot \E^{\Q}\left[\left.f_{n+1}(S_{n+1},v_{n+1}) \right. \mid \left. S_{n} = S, [v^j]_n^{n+1} \right. \right]
\end{equation*}
where $A^{N}_n = \frac{1}{N} \sum_{j=1}^{N} \phi(v_n)\phi(v_n  )^T $. Note that  $f_{n+1}$ involves the idealized continuation value at time $n+1$.

\subsubsection{Estimated Continuation Functions}

The \textit{estimated continuation functions} are the continuation functions produced from our algorithm: ${C}_n^{N}(S,v) = {a}_n^N(S) \cdot \phi_n(v)$. The regression coefficients are given by
\begin{equation}
a^{N}_n(S) =  \big[ A^{N}_n\big]^{-1}   \cdot   \frac{1}{N} \sum_{j=1}^{N}   \  \phi(v^j_n) \cdot  \E^{\Q}\left[\left. f_{n+1}^{N}(S_{n+1},v^j_{n+1} )\right.  \mid \left. S_{n} = S, [v^j]_n^{n+1} \right. \right] \
\label{eqn:estcoef}
\end{equation}
where
\begin{align*}
f^{N}_{n}(S,v) & = \max( h_n(S), C^{N}_{n}(S,v) ) \ \text{ for } \ n < M\,, \text{ and } \\
f^{N}_{M}(S,v) & = h_M(S). \
\end{align*}

\subsection{Truncation Scheme}
\label{subsection:truncationscheme}
We now state the following truncation scheme for our least-squares regression. It ensures that the coefficients produced by the algorithm are well defined and converge in a sense to be described later on.

\begin{assumption}[Truncation Conditions]
\
\label{assump:trunc}
\begin{enumerate}
\item The basis functions $\{\phi_m\}_{m=1}^{d_B}$ are bounded and supported on a compact rectangle.
\label{item:cond1}
\item The inverse of the regression matrix satisfies $\sup_{n,N} ||[A^N_n]^{-1}||_{\infty,\Q'} < \infty$ provided the inverse is defined.
\label{item:cond2}
\item For each $i=1,...,M$, the exercise values $h_{t_i}(\cdot)$ are bounded with compact support in $\calS \subset \R^{d_S}$.
\label{item:cond3}
\end{enumerate}
\end{assumption}

Condition (\ref{item:cond1}) may be imposed by limiting the support of $\{\phi_l\}$ on a bounded domain as they are typically smooth. By making $\supp \ \phi \subset \R^{d_v}$ to be a very large rectangle, the value function is essentially unaffected.

Condition (\ref{item:cond2}) is imposed by replacing $[A^N_n]^{-1}$ with $[A^N_n]^{-1} I^{n,N}_{R}$ where $I^{n,N}_{R}$ is the indicator of the event that $[A^N_n]^{-1}$ is uniformly bounded by some constant $R$. If $R > |A^{-1}_n|$ then we have $[A^N_n]^{-1} I^{n,N}_{R} \to A^{-1}_n$ $\Q'$-a.s. Again, by making $R$ a very large constant, this has essentially no effect on the values obtained by the algorithm.

Condition (\ref{item:cond3}) on the functions $h$ are always satisfied in practice as numerically solving a PDE involves truncation of $h$'s domain.
%, hence there is little harm in assuming it when analyzing the algorithm theoretically.

\begin{lemma} Given the truncation conditions, the functions $H_{n,S}$ defined in (\ref{eqn:functionH}) are finite valued for all $n = 1,...,M-1$.
\label{Hlemma}
\end{lemma}
The proof is omitted due to its simplicity. The next lemma establishes a useful relationship between the idealized, almost idealized and estimated coefficients.
\begin{lemma}
Let $n \in \{1,...,M-2\}$, $S \in \calS$. There exists a constant, $c$, which depends on our truncation conditions, such that
\begin{equation*}
|a_n^{N}(S) - a_n(S)| \leq c \cdot \beta^{N}_n(S) + \delta^{N}_n(S)
\end{equation*}
where
\begin{equation*}
\beta^{N}_n(S) = \frac{1}{N} \sum_{j=1}^{N}  \E^{\Q}\left[ \; \left. | a_{n+1}^{N}(S_{n+1}) - a_{n+1}(S_{n+1})    |    \right.  \mid  \left. S_{n} = S, [v^j]^{n+1}_n \right. \; \right] .
\end{equation*}
and
\begin{equation}
\delta^{N}_n(S) = | \widetilde{a}^{N}_n(S) - a_n(S)|
\label{eqn:deltadef}
\end{equation}
\label{lem:identity}
\end{lemma}

\begin{proof}
Given $n \in \{1,...,M-1\}$ and $S \in \mathcal{S}$ we have
\begin{equation*}
|a^{N}_n(S) - a_n(S)  | \leq | a^{N}_n(S) - \widetilde{a}^{N}_n(S)   |  + | \widetilde{a}_n^{N}(S)   - a_n(S) |.
\end{equation*}
After simplifying, we find
\begin{align}
|a^{N}_n(S) - \widetilde{a}^{N}_n(S)| &  \nonumber \\
\leq |[A^{N}_n]^{-1}_{R_A} |  \cdot & \frac{1}{N}  \sum_{j=1}^{N}   \E^{\Q}\left[ \left. | f_{n+1}^{N}(S_{n+1},v^j_{n+1}) - f_{n+1}(S_{n+1},v^j_{n+1}) |  \right. \mid \left. S_{n} = S, [v^j]_{n}^{n+1} \right. \right]   \cdot | \phi(v^j_n) | \nonumber \\
\leq c'  \cdot  \frac{1}{N}  &\sum_{j=1}^{N}   \E^{\Q}\left[ \left. | f_{n+1}^{N}(S_{n+1},v^j_{n+1}) - f_{n+1}(S_{n+1},v^j_{n+1}) |  \right. \mid \left. S_{n} = S, [v^j]_{n}^{n+1} \right. \right] \label{eqn:betasum} \
\end{align}
We then focus on the difference within the expectation. Using the inequality $|\max(a,b) - \max(a,c)| \leq |b - c|$ we find
\begin{align}
& \E^{\Q}\left[ \left. | f_{n+1}^{N}(S_{n+1},v^j_{n+1}) - f_{n+1}(S_{n+1},v^j_{n+1}) | \right. \mid \left. S_{n} = S, [v^j]^{n+1}_{n} \right.  \right]    \nonumber \\
& = \E^{\Q}\left[ \left.| \max( h_{n+1}(S_{n+1}), C_{n+1}^{N}(S_{n+1},v^j_{n+1}) ) - \max( h_{n+1}(S_{n+1}), C_{n+1}(S_{n+1},v^j_{n+1}) ) | \right.  \mid \left. S_{n} = S, [v^j]^{n+1}_{n} \right. \right] \nonumber    \\
&  \leq \E^{\Q}\left[ \left. | C^{N}_{n+1}(S_{n+1},v^j_{n+1})  - C_{n+1}(S_{n+1},v^j_{n+1})  |\right.  \mid \left. S_{n} = S, [v^j]_n^{n+1}\right. \right]    \nonumber \\
& \leq   \E^{\Q}\left[ \left. |  a^{N}_{n+1}(S_{n+1})- a_{n+1}(S_{n+1})| \right. \mid  \left. S_{n} = S, [v^j]_n^{n+1} \right. \right] \cdot  |\phi(v^j_{n+1}) | \nonumber \\
& \leq c'' \cdot  \E^{\Q}\left[ \left. |  a^{N}_{n+1}(S_{n+1})- a_{n+1}(S_{n+1})|    \right. \mid  \left. S_{n} = S, [v^j]^{n+1}_n\right. \right]  \label{eqn:finalterm}  \
\end{align}
where $c', c''$ depend on our truncation conditions. Substituting (\ref{eqn:finalterm}) into (\ref{eqn:betasum}), we obtain the result.
 \end{proof}
\subsection{Almost-Sure Convergence Under Pure SV Models}
\label{subsection:asconvergence}
In this section we prove that the coefficients also converge almost surely for a class of models with certain separability conditions that are satisfied by our SV model.

\begin{assumption}[Separability Conditions]
\
Let $n \in \{1,...,M-1\}$
	\begin{enumerate}
	\item The process $S_t$, for all $t\in[0,T]$, takes values in
	\[
	\calS = \left\{ (S^{(1)},...,S^{(d_S)}) \in \R^{d_S} \ | \ S^{(i)} > 0, \forall i = 1,...,d_S     \right\}, \qquad \text{a.s.}
	\]
	\label{item:pdegrid}
	\item If $S_n = S$ almost surely, i.e., the process begins at $S$ at time $t_n$, then $S_{n+1} = S \odot R_n  = (S_n^{(1)}R_n^{(1)},..., S_n^{(d_s)}R_n^{(d_S)} )$ where $R_n$ does not depend on the value of $S_n$. $R_n$ takes values in $\calS$ a.s.
	\label{item:return}
	\item The exercise function, $h$, is continuous with compact support in $\calS$. The basis functions $\phi$ are compactly supported and continuous on $\R^{d_v}$.
	\label{item:ctny}
\end{enumerate}
	\label{assump:separabilityconditions}
\end{assumption}
%\seb{put in some small comments about conditions 1 and 2... I know you have something below about 2, but add something without formula} \david{done}

Condition (\ref{item:pdegrid}) limits our analysis to assets which take only positive values such as equities and foreign exchange rates.

Condition (\ref{item:return}) allows us to separate our future asset price as a product of its current price and return. The assumption that $R_n$ take values in $\calS$ implies that they are finite valued a.s. As a result, letting $\e>0$, there exists tuples $\{(r^{(i)}_{l},r^{(i)}_{h})\}_{i=1}^{d_S}$, where $r^{(i)}_l > 0$ such that
\begin{equation}
\Omega_{r_l,r_h}^{\e} = \left\{ \  r^{(i)}_{l} \leq R^{(i)}_n \leq r^{(i)}_{h} \mid \forall \ n \in \{1,\ldots,M-1\} , \ \forall \ i \in \{1,\dots,d_S\} \  \right\} \subset \Omega
\label{eqn:posquad1}
\end{equation}
satisfies $\Q(\Omega_{r_l,r_h}^{\e}) > 1 - \e.$ We also write
\begin{equation*}
E^{\e}_{r_l,r_h} = \left \{  R \in \R^{d_S}  \mid  r^{(i)}_l \leq R^{(i)} \leq r^{(i)}_h, \forall \ i \in \left\{1,\ldots,d_S\right\} \   \right\}.
\label{eqn:posquad2}
\end{equation*}
Given $E_{r_l,r_h}^{\e}$, we may find an open set $U^{\e}$ such that $E_{r_l,r_h}^{\e} \subset U^{\e}$ and $ \overline{U^{\e}} \subsetneq \calS$. By Urysohn's Lemma/Tietze's Extension theorem (see \cite{Munkres}), there exists a map $\eta_{E} : \calS \to \R$ such that $\eta_{E} = 1$ on $E^{\e}_{r_l,r_h}$, $\eta_{E} = 0$ on $ \calS \setminus U^{\e}$ and $||\eta_E||_{\infty} \leq 1$, i.e. a bump function supported on $E$. In most cases, our notation will suppress dependence on $r_l, r_h, \e$.

Condition (\ref{item:ctny}) allows us to apply the Stone-Weierstrass (SW) theorem which underlies the `separation technique' that will be demonstrated in upcoming Lemmas. We apply the version of SW for functions on unbounded domains that vanish at infinity (See \cite{Folland}). Suppose we are given the payoff function for a call option, i.e. $g : \R \to \R$ where $g(x) = (x-K)_+$. To modify $g$ such that it falls within our assumption, we first truncate its support to obtain a function $f(x) = (x-K)_+ I_{(0,R_1)}(x)$ where $R_1$ is large number. Finally, we continuously extend $f$ on $\R$ such that $f = 0$ on $(R_2,\infty)$ where $R_2 > R_1$. A similar construction may be done for a put option payoff near 0, and payoffs on higher dimensional domains.

Under these assumptions, we have the following main result.

\begin{theorem}
Let $n \in \left\{1,\ldots,M-1\right\}$ and $S \in \calS$ be fixed. Then we have
\begin{equation*}
\lim_{{N}\to \infty} | a_n^{N}(S) - a_n(S) | = 0
\end{equation*}
$\Q'$-almost surely
\label{thm:as}
\end{theorem}
Theorem \ref{thm:as} tells us that for almost every choice of sequence of paths, our coefficients based on these paths converge to the true idealized coefficient. It then follows, since $C^N_n(S,v) = a^N_n(S) \cdot \phi(v)$, that the continuation values, option prices, and optimal exercise boundaries also converge almost surely (up to our basis assumption's bias).

The proof borrows ideas from \cite{TV} and \cite{Clement} and takes the following steps
\begin{enumerate}
\item \underline{Lemma \ref{lemma:geom}}. Carry out a geometric construction that allows us to approximately separate functions $h$ of $S_{t_n}$ that are continuous and compactly supported in $\calS$. The function that provides the approximate separation is denoted as $\psi$.
\item \underline{Lemma \ref{lemma:finalapprox}}. Use the geometric construction to show the explicit relationship between $h$ and the separating functions, $\psi$ and thus demonstrate what is referred to as a separating estimate.
\item \underline{Lemma \ref{lemma:mm1}}. Prove the theorem for $n = M-1$ and also obtain an almost-sure separating estimate for $|a^N_{M-1}(S) - a_{M-1}(S)|$.
\item \underline{Lemma \ref{lemma:mm2}}. Prove the theorem for $n = M-2$ and also obtain an almost-sure separating estimate for $|a^N_{M-2}(S) - a_{M-2}(S)|$. The separating estimate for $n = M-2$ involves the function $\delta^N_{M-2}(S)$ as defined in (\ref{eqn:deltadef}).
\item \underline{Lemma \ref{lemma:deltaest}}. Develop an almost-sure separating estimate for $\delta^N_n$ for all $n \in \{1,\ldots,M-2\}$.
\item \underline{Proposition \ref{prop:main} }. Prove the theorem for $n = \{1,\ldots,M-3\}$ using Lemma \ref{lemma:mm2} and lemma \ref{lemma:deltaest}. Also obtain an almost-sure separating estimate for $|a^N_{n}(S) - a_{n}(S)|$ which is used during the induction.
\end{enumerate}

\begin{lemma}
Let $\e > 0$ and $h: \calS \to \R$ be continuous and compactly supported in $\calS$. Let $\widetilde{h}: \calS \times \calS \to \R$ be defined via
\begin{equation*}
\widetilde{h}(S,R ) = h(S \odot R) \cdot \eta_{E}(R)
\end{equation*}
where $E$ results from (\ref{eqn:posquad1}). There exists a map $\psi : \calS \times \calS \to \R$ of the form
\begin{equation*}
\psi(S,R) = \sum_{i=1}^k \psi_{i,S}(S)\psi_{i,R}(R)
\end{equation*}
such that $\psi_{i,S}, \psi_{i,R} \in C_c(\calS)$  and $|| \widetilde{h} - \psi ||_{\infty} < \e$
\label{lemma:geom}
\end{lemma}

\begin{proof}

It follows from the properties of the mapping $i(S,R) = S \odot R$ and compact support of $\eta_E$ and $h$, that $\widetilde{h}$ is compactly supported in $\calS \times \calS$. We also have that $\widetilde{h}$ is continuous on $\calS \times \calS$.

To construct $\psi$, we begin by defining the algebra of functions
\begin{equation*}
\calA := \left\{ \sum_{i=1}^k \psi_{i,S}(S) \psi_{i,R}(R) \ | \ \psi_{i,S}, \psi_{i,R} \in C_c(\calS), \ k \in \N    \right\}
\end{equation*}
and show $\calA$ is dense in $C_0(\calS\times\calS)$ under the uniform metric.

Given two distinct points $\{(S_i,R_i)\}_{i=1}^2$, without loss of $  $generality, we assume $S_1 \neq S_2$. We now find bump functions $\eta_{S,1}, \eta_{S,2} \in C_c(\calS)$ that separate $S_1, S_2$, i.e., $\eta_{S,1}(S_1) = 1, \eta_{S,1}(S_2) = 0$, $\eta_{S,2}(S_1) = 0, \eta_{S,2}(S_2) = 1$, along with a bump function $\eta_{R} \in C_c(\calS)$ that is supported on $R_1$ and $R_2$. Letting $\psi_{1}(S,R) = \eta_{S,1}(S) \cdot \eta_{R}(R)$ and $\psi_2(S,R) = \eta_{S,2}(S) \cdot \eta_{R}(R)$, we have that $\calA$ separates points. Since $\calA$ contains bump functions supported at each point $(S,R) \in \calS\times\calS$, it vanishes nowhere.
\end{proof}

\begin{lemma}
Let $\e>0$ and $h : \calS \to \R $ be continuous and supported in $\calS$. Let $\widetilde{h}$ be as in the statement of Lemma \ref{lemma:geom} and $\psi(S,R)$ be a separable $\e$-approximation of $\widetilde{h}$. There exists a random variable $F(S)$ and function $G(S) : \calS \to \R$ such that
\begin{align*}
\E^{\Q}\left[ \phi(v_n) h(S_{n+1}) \mid S_n = S, [v]_n^{n+1} \right] &= F(S) +  \phi(v_n)\sum_{i=1}^k \psi_{S,i}(S)  \E^{\Q}\left[ \psi_{R,i}(R_n) \mid [v]_n^{n+1}  \right] \\
 \ \E^{\Q}\left[ \phi(v_n) h(S_{n+1})  \mid S_n = S \right] \  &= G(S) + \sum_{i=1}^k \psi_{S,i}(S) \E^{\Q}\left[ \  \phi(v_n) \psi_{R,i}(R_n)  \right]
\end{align*}
where $||F(S)||_{\infty,\Q'} \leq c \ \e, \  |G(S)| \leq c' \ \e$ for all $S \in \calS$ where $c,c'$ depend only on our truncation and separability conditions.
\label{lemma:finalapprox}
\end{lemma}

\begin{proof}
We will only show the first relation as they are similar.
\begin{align*}
&  \hspace*{-1em} \E^{\Q} \left[ \phi(v_n) h(S_{n+1}) \mid S_n = S, [v]_n^{n+1} \right]  \\
&  =  \ \E^{\Q}[ \phi(v_n)  h(S \odot R_n) (1-\eta_E(R_n) )   \ | \ S_n = S, [v]_n^{n+1} ] + \E^{\Q}\left[ \phi(v_n)  h(S \odot R_n) \eta_{E}(R_n)    \mid S_n = S, [v]_n^{n+1} \right]    \\
 & =  \   F_1(S) +  \E^{\Q}\left[ \phi(v_n)  \widetilde{h}(S \odot R_n) \mid S_n = S, [v]_n^{n+1} \right]    \\
& =  \   F_1(S) +  \E^{\Q}\left[ \phi(v_n)  (  \widetilde{h}(S \odot R_n) - \psi(S,R_n) )  \mid S_n = S, [v]_n^{n+1} \right]  + \E^{\Q}\left[ \phi(v_n)  \psi(S,R_n) \mid S_n = S, [v]_n^{n+1} \right]       \\
& =  \ F_1(S)     +  F_2(S) + \E^{\Q}\left[ \phi(v_n)  \psi(S,R_n) \mid S_n = S, [v]_n^{n+1} \right]   \\
& =  \ F_1(S)     +  F_2(S) +  \E^{\Q}\left[ \phi(v_n)  \psi(S,R_n) \mid S_n = S, [v]_n^{n+1} \right]   \\
& =  \ F_1(S)     +  F_2(S) +  \phi(v_n)\sum_{i=1}^k \psi_{S,i}(S)  \E^{\Q}\left[ \psi_{R,i}(R_n) \mid [v]_n^{n+1}  \right] \
\end{align*}
It then follows that $||F_i(S)||_{\infty,\Q'} < c_i \ \e$ for all $S \in \calS$ and $i = 1,2$, where $c_i$ depend on our truncation and separability conditions. Finally, we set $F(S) = F_1(S) + F_2(S)$.
\end{proof}

\begin{lemma}
Let $\e>0$. We have
\label{lemma:sepest}
\begin{equation*}
|a^{N}_{M-1}(S) - a_{M-1}(S)| \leq c_{M-1} \  \e +  \sum_{i_{M-1} = 1}^{k_{M-1} } \alpha^{i_{M-1},N}_{M-1,M} \ \psi_{S,i_{M-1}}(S) \\
\label{eqn:mclaim}
\end{equation*}
a.s., where $\alpha^{i_{M-1},N}_{M-1,M}$ are random variables that depend on $\{ [v^j]^{M}_{M-1} \}_{j=1}^{\infty}$ and $\lim_{N\to \infty} \alpha^{i_{M-1},N}_{M-1,M}= 0$ a.s., $\psi_{S,i_{M-1}}$ satisfy the conditions of Lemma \ref{lemma:finalapprox} and $c_{M-1}$ depends on our truncation and separability conditions.
 \label{lemma:mm1}
 \end{lemma}

\begin{proof}

\

Given $\e > 0$ we have
\begin{align*}
|a^{N}_{M-1}(S) - a_{M-1}(S)| &=| \ [A^{N}_{M-1}]^{-1} \frac{1}{N} \sum_{j=1}^{N} \phi(v^j_{M-1})\E^{\Q}[ h(S_M) \ | \ [v^j]_{M-1}^M, S_{M-1}=S  ] \\
& \ \ \ \ - A_{M-1}^{-1}  \E^{\Q}[ h(S_M) \  \phi(v_{M-1}) \ | \ S_{M-1}=S  ] \  | .\
\end{align*}
By Lemma \ref{lemma:finalapprox} we can find a constant $c$ and a separable $\e$-approximation, $ \psi_{M-1}(S,R)$ such that
\begin{align*}
|a^{N}_{M-1}(S) - a_{M-1}(S)| & \leq | \  \left[A^{N}_{M-1}\right]^{-1} \frac{1}{N} \sum_{j=1}^{N} \phi(v^j_{M-1}) \sum_{i_{M-1} = 1}^{k_M} \psi_{S,{i_{M-1}  }}(S) \E^{\Q}\left[  \psi_{R,i_{M-1}  }(R_{M-1})  \mid [v^j]_{M-1}^M \ \right]  \\
& \ \ \ - A_{M-1}^{-1} \sum_{i_{M -1}=1}^{k_M} \psi_{S,{i_{M-1}}} (S)\E^{\Q}\left[ \phi(v_{M-1}) \psi_{R,i_{M-1}  }(R_{M-1}) \ \right] \ | + c \ \e \\
& \leq  \sum_{i_{M -1 } =1}^{k_{M-1}  } |\psi_{S,i_{M -1} }(S)| \cdot \alpha^{i_{M-1},N}_{M-1,M} +  c \ \e \
\end{align*}
where the last line follows from interchanging the summations, $c$ depends on our truncation and separability conditions and
\begin{align*}
\alpha^{i_{M -1} ,N}_{M-1,M} &= | \  [A^{N}_{M-1}]^{-1} \frac{1}{N} \sum_{j=1}^{N} \phi(v^j_{M-1})\E^{\Q}\left[ \psi_{R,i_{M-1}  }(R_{M-1})  \mid [v^j]_{M-1}^M \ \right] \\
& \ \ \ \   - A_{M-1}^{-1} \E^{\Q}\left[  \phi(v_{M-1}) \psi_{R,i_{M-1}  }(R_{M-1}) \  \right] \ | .
\label{eqn:alphamm1}
\end{align*}
By the SLLN, $\lim_{N\to\infty} \alpha^{i_{M-1},N}_{M-1,M} = 0 $ a.s. for all $i_{M-1} \in \{1,...,k_M\}$.
\end{proof}

\begin{lemma}
Let $\e > 0$. We have
\begin{align*}
|a^{N}_{M-2}(S) - a_{M-2}(S)| & \leq c^N_{M-2} \ \e +  \sum_{i_{M-1} = 1}^{k_{M-1} } \alpha^{i_{M-1},N}_{M-1,M}\sum_{i_{M-2}}^{k_{M-2}} \alpha^{i_{M-2},i_{M-1},N}_{M-2,M-1}  \psi_{S,i_{M-2},i_{M-1}}(S) \\
& + \delta^N_{M-2}(S)\
\end{align*}
a.s. where $ \alpha^{i_{M-1},N}_{M-2,M-1}$ are random variables that depend on $\{ [v^j]_{M-2}^{M-1}\}_{j=1}^{\infty}$ and satisfy
\begin{equation*}
\limsup_{N\to\infty} \alpha^{i_{M-2},N}_{M-2,M-1} < \infty.
\end{equation*}
Also, $\limsup_{N\to\infty} c^N_{M-2} < \infty $ with a bound depending only on our truncation and separability conditions and $\psi_{S,i_{M-1},i_{M-2}}$ satisfy the conditions of Lemma \ref{lemma:finalapprox}. Lastly, $\delta^N_{M-2}$ is as defined in (\ref{eqn:deltadef}).
\label{lemma:mm2}
\end{lemma}

\begin{proof}
Given $\e > 0$, and using Lemmas \ref{lem:identity} and \ref{lemma:mm1} we have a.s.
\begin{align*}
|a^{N}_{M-2}(S) - a_{M-2}(S) | &\leq \widetilde{c} \frac{1}{N} \sum_{j=1}^{N}  \E^{\Q}\left[  \left. |a^{N}_{M-1}(S_{M-1}) - a_{M-1}(S_{M-1}) |   \right.  \mid  \left. S_{M-2} = S, [v^j]^{M-1}_{M-2} \right. \right]  +  \delta^{N}_{M-2}(S) \\
& \leq c' \e + \sum_{i_{M-1}}^{k_{M-1}} \alpha^{i_{M-1},N}_{M-1,M} \frac{1}{N}\sum_{j=1}^N \E^{\Q}\left[ \psi_{S,i_{M-1}}(S_{M-1}) \mid S_{M-2} = S  \ [v^j]_{M-2}^{M-1}   \right] \\
& + \delta^{N}_{M-2}(S) \
\end{align*}
For each $\psi_{S,i_{M-1}}$, we apply Lemma \ref{lemma:finalapprox}, and obtain a separable $\e$-approximating function $\psi_{i_{M-1}}(S,R)$ of the form
\begin{equation*}
\psi_{i_{M-1}}(S,R) = \sum_{i_{M-2}=1}^{k_{M-2}} \psi_{S,i_{M-2},i_{M-1}}(S)\psi_{R,i_{M-2},i_{M-1}}(R)
\end{equation*}
and apply it to the expectation. This results in
\begin{align*}
|a^N_{M-2}(S) - a_{M-2}(S) | &\leq (c' + \widetilde{c}^{N}_{M-2}) \e + \sum_{i_{M-1} = 1}^{k_{M-1} } \alpha^{i_{M-1},N}_{M-1,M}\sum_{i_{M-2}=1}^{k_{M-2}} \alpha^{i_{M-2},i_{M-1},N}_{M-2,M-1}  \psi_{i_{M-2},i_{M-1}}(S) + \delta^N_{M-2}(S) \
\end{align*}
where $\lim_{N\to \infty} \widetilde{c}^N_{M-2} = 0$ and $\alpha^{i_{M-2},N}_{M-2,M-1} = \frac{1}{N} \sum_{j=1}^N \E^{\Q}\left[ \;  \psi_{R,i_{M-2},i_{M-1}}(R_{M-2}) \mid [v^j]_{M-2}^{M-1} \; \right] $.
\label{eqn:alphamm2} By the SLLN, $\delta^N_{M-2}(S)$  and  $\alpha^{i_{M-2},N}_{M-2,M-1}$  for each $(i_{M-1},i_{M-2})$ converge to 0 a.s. Finally, setting $c^N_{M-2} = c' + \widetilde{c}^N_{M-2}$ implies the lemma.
\end{proof}

\begin{lemma}
Let $n \in \{1,...,M-2 \}$, $ \delta^{N}_n(S)$ as in (\ref{eqn:deltadef})  and $\e > 0$. We have
\begin{equation*}
\delta^N_{n}(S) \leq c \cdot \e + \sum_{i_n=1}^{k_n} \beta^{i_n,N}_{n,n+1} \psi_{i_n}(S)
\end{equation*}
$\Q'$-a.s. where $\beta^{i_n,N}_{n,n+1}$ are random variables that depend on $\{ [v^j]_n^{n+1} \}_{j=1}^{\infty}$, $\lim_{N\to\infty} \beta^{i_n,N}_{n,n+1} = 0$ a.s., $c$ depends on our truncation and separability conditions, and $\psi_{i_n}$ satisfy the conditions of Lemma \ref{lemma:finalapprox}.
\label{lemma:deltaest}
\end{lemma}

\begin{proof}
We break the proof up into multiple stages.

\textbf{Preliminary estimates on} $a_k, f_k$.

Let $f_k$ be as defined in (\ref{eqn:functionf}). We begin by showing that for each $j \in \{n,...,M-1\}$ that
\begin{align*}
a_j(S) &= F_j(S) + G_j(S) \\
f_j(S,v) &= H_j(S,v) + J_j(S,v)  \
\end{align*}
where $|F_k(S)| \leq c_{k} \e$, $|H_k(S,v)| \leq c'_{k} \e$ with $c_{k},c'_k$ depending on our truncation conditions and separability conditions. The functions $G_j$ admit the representation $G_j(S) = \sum_{i_j =1}^{k_j} c_{i_j, j} \psi_{S,i_j}(S) $ where $c_{i_j,j} \in \R^{d_B} $ and $\psi_{S,i_j} \in C_c(\calS)$. Finally $J_j(S,v) = \max( h(S),G_{j}(S) \cdot \phi(v))$.
To this end, we let $j = M-1$ and have that
\begin{align}
a_{M-1}(S) =& A^{-1}_{M-1} \E^{\Q}\left[ \phi(v_{M-1}) h(S_M)  \mid S_{M-1} = S  \right] \nonumber \\
       =& A^{-1}_{M-1} \E^{\Q}[ \phi(v_{M-1}) h(S \odot R_{M-1})  (1-\eta_{E}(R_{M-1}))  \   ] \nonumber \\
       & + A^{-1}_{M-1} \E^{\Q}\left[ \phi(v_{M-1}) h(S \odot R_{M-1})  \eta_{E}(R_{M-1})  \right]  \
\label{eqn:am1}
\end{align}
where $E$ is as in (\ref{eqn:posquad2}) and focus on the second term. We let
\begin{align*}
\widetilde{\Xi}_{M-1} & : \R^{2d_S + d_v } \to \R^{d_B} \\
\widetilde{\Xi}_{M-1}(S,R,v_1) & =  \phi(v_1)  h(S \odot R) \eta_E(R) \\
                                   & =  \phi(v_1) \widetilde{h}(S ,R) \
\end{align*}
and apply Lemma \ref{lemma:geom} to find an $\e$-separating function of the form  $\psi_{M-1}(S,R) = \sum_{i_{M-1} = 1}^{k_{M-1}} \psi_{S,i_{M-1}}(S)\psi_{R,i_{M-1}}(R)$ where $\psi_{S,i_{M-1}}, \psi_{R,i_{M-1}} \in C^0_c(\calS)$ and $||\widetilde{h}-\psi_{M-1}||_{\infty} < \e$.
This allows us to write
\begin{equation*}
\widetilde{\Xi}_{M-1}(S,R,v_1) =  F_{2,M-1}(S,R,v_1) + \phi(v_1)\psi_{M-1}(S,R)
\end{equation*}
where
$ F_{2,M-1}(S,R,v_1) = \phi(v_1)(\widetilde{h}(S,R) - \psi_{M-1}(S,R) )$ and $
\E^{\Q}[ F_{2,M-1}(S,R_{M-1},v_{M-1}) ]  < c \ \e$ with $c$ depending on our truncation and separability conditions.

%F_{1,M-1}(S,R,v_1) +

%\begin{equation*}
% F_{1,M-1}(S) = A^{-1}_{M-1} \E^{\Q}[ \phi(v_{M-1}) h(S \odot R_{M-1})  (1-\eta_{E}(R_{M-1}))  \   ],
%\end{equation*}

Returning to our expression in $(\ref{eqn:am1})$, we find
\begin{align*}
a_{M-1}(S) &= F_{M-1}(S) + \sum_{i_{M-1} =1}^{k_{M-1}} c_{i_{M-1},M-1} \psi_{S,i_{M-1} }(S) \\
           &=: F_{M-1}(S) + G_{M-1}(S) \
\end{align*}
where $|F_{M-1}(S)| \leq c_{M-1} \  \e$ for all $S \in \calS$ , $c_{M-1}$ depends on our truncation and separability conditions, and
\begin{equation*}
c_{i_{M-1},M-1} = \E^{\Q}\left[  \phi(v_{M-1}) \psi_{R,i_{M-1}}(R_{M-1})   \right].
\end{equation*}
We now turn to $f_{M-1}(S,v)$
\begin{align*}
f_{M-1}(S,v) &= \max( h(S) , a_{M-1}(S)\cdot \phi(v) ) \\
&= \max(h(S)  , F_{M-1}(S)\cdot \phi(v) + G_{M-1}(S) \cdot \phi(v) )\\
&= \Big[\max(h(S), F_{M-1}(S)\cdot \phi(v) + G_{M-1}(S) \cdot \phi(v) ) - \max(h(S), G_{M-1}(S)\cdot \phi(v) ) \Big] \\
&+ \max(h(S), G_{M-1}(S) \cdot \phi(v) ). \
\end{align*}
Equating $H_{M-1}(S,v)$ to the term in brackets and $J_{M-1}(S,v) := \max( h(S),G_{M-1}(S)\cdot \phi(v))$,
gives us the required form and concludes the claim for $j = M-1$.

Now let $j \in \{1,...,M-2\}$, we have
\begin{align*}
a_j(S) & = A^{-1}_{j} \E^{\Q}\left[ \phi(v_j)f_{j+1}(S_{j+1},v_{j+1}) \right] \\
       & = A^{-1}_j \E^{\Q}\left[ \phi(v_j)H_{j+1}(S_{j+1},v_{j+1}) \right] \\
       & + A^{-1}_j \E^{\Q}\left[ \phi(v_j) \max(h(S\odot R_{j} ), G_{j+1}(S\odot R_j) \cdot \phi(v_{j+1})   ) (1 - \eta_{E_{j+1}}(R_n) )  )  \right]  \\
       & + A^{-1}_j \E^{\Q}\left[ \phi(v_j) \max(h(S\odot R_j), G_{j+1}(S\odot R_j) \cdot \phi(v_{j+1})   )  \eta_{E_{j+1} }(R_n)     \right],  \
\end{align*}
where $E_{j+1}$ is as in (\ref{eqn:posquad1}) so that $\Q(\Omega^{\e_{j+1},c}) < \frac{\e}{||J_{j+1}||_{\infty} }$, and focus on the final term. By assumption, $G_{j+1}(S) = \sum_{i_{j+1}=1}^{k_{j+1}} c_{i_{j+1} } \psi_{S,i_{j+1}  }(S)$ where $c_{i_{j+1} }  \in \R^{d_B}$. Letting $d_{i_{j+1}  }(v) = c_{i_{j+1} }\cdot\phi(v)$, we are led to consider the function
\begin{align*}
\widetilde{\Xi}_{j} &: \R^{2d_S + 2d_v} \to \R^{d_B} \\
\widetilde{\Xi}_{j}(S,R,v_1,v_2) &= \phi(v_1) \max\left( h(S\odot R)\eta_{E_{j+1}}(R), \sum_{i_{j+1}=1}^{k_{j+1}} d_{i_{j+1}}(v_2) \psi_{S,i_{j+1}}(S\odot R)\eta_{E_{j+1}}(R)  \right) \\
&= \phi(v_1) \max\left( \widetilde{h}(S,R)  , \sum_{i_{j+1}=1}^{k_{j+1}} d_{i_{j+1}}(v_2) \widetilde{\psi}_{S,i_{j+1}}(S,R)  \right) \
\end{align*}
where $\widetilde{\psi}_{S,i_{j+1}}(S,R) = \psi_{S,i_{j+1}}(S,R)\eta_{E_{j+1}}(R)$.
%and $E_{j+1} \subset \R^{d_S}$ induces $\Omega^{\e_{j+1} } \subset \Omega$ as in (\ref{eqn:posquad1})
Before carrying out a construction very similar to the $j = M-1$ case, we let $A_{j+1}  =   ||\widetilde{h}||_{\infty} + \sum_{i_{j+1} = 1}^{k_{j+1}} ||d_{i_{j+1}}||_{\infty} ||\widetilde{\psi}_{S,i_{j+1}} ||_{\infty}.$ Next, we note that
\begin{equation*}
\max\left( \widetilde{h}   , \sum_{i_{j+1}=1}^{k_{j+1}} d_{i_{j+1}} \widetilde{\psi}_{S,i_{j+1}}  \right) \in C_c\left(\calS \times \left( \calS \times \R^{d_v} \right) \right)
\end{equation*}
and find a function of the form $ \psi_j(S,R,v) = \sum_{i_{j}=1}^{k_{j}  } \psi_{S,i_j}(S)\psi_{R,v,i_j}(R,v)$ such that $\psi_{S,i_j} \in C_c(\calS)$, $\psi_{R,v,i_j} \in C_c(\calS \times \R^{d_v})$ and
\begin{equation*}
\left|\left|  \max \left( \widetilde{h} , \sum_{i_{j+1} =1}^{k_{j+1} } d_{i_{j+1}} \widetilde{\psi}_{S,i_{j+1} }  \right) - \psi_j \right|\right|_{\infty} < \e.
\end{equation*}
Finally, we write $a_j(S) = F_{j,1}(S) + F_{j,2}(S) + F_{j,3}(S) + \sum_{i_j=1}^{k_j} c_{i_j,j} \psi_{S,i_j}(S)$ where
\begin{align*}
F_{j,1}(S) &= A^{-1}_j \E^{\Q}\left[ \phi(v_j)H_{j+1}(S_{j+1},v_{j+1}) \ | \ S_j = S \right],\\
F_{j,2}(S) &= A^{-1}_j \E^{\Q}\left[ \phi(v_j) \max(h(S\odot R_{j} ), G_{j+1}(S\odot R_j) \cdot \phi(v_{j+1})   ) (1 - \eta_{E_{j+1}}(R_n) )  )  \right],  \\
F_{j,3}(S) &= \E^{\Q}\left[ \ \widetilde{\Xi}_j(S,R_j,v_j,v_{j+1}) - \phi(v_j) \psi_j(S,R_j,v_{j+1})   \ \right]\
\end{align*}
so that for  $i=1,2,3$, we have $||F_{j,i}||_{\infty} \leq c'_{j,i} \e  $ where $c'_{j,i}$ depend on our truncation and separability conditions and $c_{i_j,j}  = \E\left[  \phi(v_j)\psi_{R,v,i_j}(R_j,v_{j+1} )  \right].$ Also we have that $\psi_{S,i_j }(S) \in C_c(\calS) $ . Writing $F_{j} = F_{j,1} + F_{j,2} + F_{j,3}$ and $G_j(S)$ to be the final term completes the claim for $a_j(S)$. Showing the result for $f_{j}(S,v)$ is analogous to the base case and so we omit the proof.

\textbf{Estimates of $\delta^N_n(S)$}

We write
\begin{align*}
\delta^N_n(S) &= | \ \left[A^N_n\right]^{-1} \frac{1}{N}\sum_{j=1}^N \E^{\Q}\left[ \phi(v_n) f_{n+1}(S_{n+1},v_{n+1})   \mid S_n = S, [v^j]_n^{n+1} \right] \\
& - A_n^{-1}\E^{\Q}\left[ \phi(v_n) f_{n+1}(S_{n+1},v_{n+1})   \mid S_n = S \right] \ | \\
& \leq | \ \left[A^N_n\right]^{-1} \frac{1}{N}\sum_{j=1}^N \E^{\Q}\left[ \phi(v_n) J_{n+1}(S_{n+1},v_{n+1})  \mid S_n = S, [v^j]_n^{n+1} \right] \\
&- A_n^{-1}\E\left[ \phi(v_n) J_{n+1}(S_{n+1},v_{n+1})  \mid S_n = S \right] \
+ c \ \e \\
& \leq | \ \left[A^N_n\right]^{-1} \frac{1}{N}\sum_{j=1}^N \phi(v^j_n)\E^{\Q}\left[  \max( h(S_{n+1}) ,G_{n}(S_{n+1}) \cdot \phi(v_{n+1})) \eta_E(R_n) \mid S_n = S, [v^j]_n^{n+1} \right] \\
&- A_n^{-1}\E^{\Q} \left[ \phi(v_n)  \max( h(S_{n+1}),G_{n}(S_{n+1}) \cdot \phi(v_{n+1})) \eta_E(R_n)   \mid S_n = S \right] \ |  \
+ c' \ \e \
\end{align*}
where $c$ and $c'$ depend on our truncation conditions. We now focus on the expression within our expectations and define the function
\begin{align*}
\widetilde{\Xi}'_{n}             &: \R^{2d_S  + 2d_v}  \to \R^{d_B} \\
\widetilde{\Xi}'_n(S,R,v_1,v_2)  & = \phi(v_1) \max( h(S \odot R),G_{n}(S \odot R) \cdot \phi(v_2)) \ \eta_E(R) \\
                                 & = \phi(v_1) \max( h(S \odot R)  \ \eta_E(R),G_{n}(S \odot R) \eta_E(R) \cdot \phi(v_2) ) \\
					             &  = \phi(v_1) \max( \widetilde{h}(S,R)  ,\widetilde{G}_{n}(S,R) \cdot \phi(v_2)\
\end{align*}
Applying techniques that are exactly analogous to previous steps, we obtain a separating estimate for $\widetilde{\Xi}'$ of the form $\widetilde{\Xi}'_n(S,R,v_1,v_2) = F(S,R,v_1,v_2) + \phi(v_1)\sum_{i_n=1}^{k_n} \psi_{S,i_n}(S)\psi_{R,v,i_n}(R,v_2)$ where $F$ is appropriately bounded and $\psi_{S,i_n} \in C^0_c(\calS)$. This leads to $\delta^N_n(S) \leq  \ c \ \e  + \sum_{i_n=1}^{k_n} |\psi_{S,i_n}(S)| \cdot \beta^{i_n,N}_{n,{n+1}}$ where $c$ again depends on our truncation conditions and
\begin{equation*}
\beta^{l,N}_{n,n+1} = | \  \left[A^N_n\right]^{-1} \ \frac{1}{N} \sum_{j=1}^N \phi(v^j_n) \E^{\Q}\left[ \psi_{R,v,i_n}(R_n,v_{n+1}) \mid [v^j]_n^{n+1} \ \right] -  A^{-1}_n \ \E^{\Q}\left[ \phi(v^j_n) \psi_{R,v,i_n}(R_n,v_{n+1}) \ \right] \ | .
\end{equation*}
By the SLLN, for each $i_n$, we have that $\lim_{N\to\infty} \beta^{i_n,N}_{n,n+1} = 0$ a.s completing the proof.
\end{proof}

\begin{proposition}
Let $n \in \{1,...,M-3\}$ and $\e>0$. We have
\label{lemma:sepest}

\begin{equation}
|a^{N}_n(S) - a_n(S)|   \leq c_n^{N} \ \e + \delta^N_n(S) + \alpha^N_n(S) + \sum_{l=n+1}^{M-2} \beta^N_{l,n}(S)
\label{eqn:main}
\end{equation}
a.s. where
\begin{equation*}
\alpha_n^N(S) = \sum_{i_{M-1} = 1}^{k_{M-1} } \alpha^{i_{M-1},N}_{M-1,M} \sum_{i_{M-2}=1}^{k_{M-2}} \alpha^{i_{M-2},i_{M-1},N}_{M-2,M-1} \ \dots \  \sum_{i_n=1}^{k_n} \alpha^{i_n,\dots,i_{M-1}, N}_{i_n,i_{n+1}} \cdot \psi^{n,\alpha}_{i_n,\dots,i_{M-1}}(S),
\end{equation*}
\begin{equation*}
\beta_{l,n}^N(S) =    \sum_{i_{l} = 1}^{m_l } \beta^{l,i_{l},N}_{l-1,l} \sum_{i_{l-1}=1}^{m_{l-1}} \beta^{l,i_l,i_{l-1},N}_{l-2,l-1} \ \dots \  \sum_{i_{n}=1}^{m_{n}} \beta^{l,i_l,\dots,i_{n}, N}_{n,n+1} \cdot \psi^{n,\beta}_{i_{l},\ldots,i_n}(S),
\end{equation*}
\begin{align*}
\lim_{N\to \infty} \alpha^{i_{M-1},N}_{M-1,M} = 0, & \ \limsup_{N} \alpha^{ i_j,\ldots,i_l, N}_{k-1,k}  < \infty, \ \forall \ k \in \{n,...,M-1\}  \
\end{align*}
a.s. and
\begin{align*}
 \forall \ l \in \{ n+1,\ldots,M-2\}, & \; \forall j \in \{ n,\ldots, l  \}  \\
 \lim_{N\to \infty} \beta^{l,i_l,N}_{l-1,l} = 0, & \ \limsup_{N} \beta^{ l,i_j,\ldots,i_l, N}_{k-1,k} < \infty \  \
\end{align*}
almost-surely. The bounds on $\limsup_{N} c^N_n$ depend on our truncation and separability conditions and $\psi^{k,\alpha}_{i_l,\ldots ,i_n },\psi^{k,\beta}_{i_l,\ldots ,i_n}$ satisfy the conditions of Lemma \ref{lemma:geom}.

\label{prop:main}
\end{proposition}

By taking limits on both sides, and noting the properties of the random variables, mappings, and constants described, the main theorem then follows.

\begin{proof}[Proof of Proposition]
We begin by letting $n = M-3$ and by Lemma \ref{lem:identity} and Lemma \ref{lemma:mm2} we have almost-surely
\begin{align}
&|a^{N}_{M-3}(S) - a_{M-3}(S) | \nonumber \\
&\leq \delta^{N}_{M-3}(S) + \widetilde{c} \frac{1}{N} \sum_{j=1}^{N}  \E^{\Q}\left[  \left. |a^{N}_{M-2}(S_{M-2}) - a_{M-2}(S_{M-2}) |   \right.  \mid  \left. S_{M-3} = S, [v^j]^{M-2}_{M-3} \right. \right] \nonumber \\
&\leq  \delta^N_{M-3}(S) + c^{N}_{M-2}  \nonumber \\
& + \widetilde{c}\sum_{i_{M-1} = 1}^{k_{M-1} } \alpha^{i_{M-1},N}_{M-1,M}\sum_{i_{M-2}}^{k_{M-2}} \alpha^{i_{M-2},i_{M-1},N}_{M-2,M-1} \frac{1}{N}\sum_{j=1}^N \E^{\Q}\left[ \psi_{S,i_{M-2},i_{M-1}}(S_{M-2}) \mid S_{M-3} = S, [v^j]_{M-3}^{M-2}  \right] \label{eqn:C2} \\
& + \widetilde{c} \frac{1}{N} \sum_{j=1}^{N} \E^{\Q}[ \delta^N_{M-2}(S_{M-2})\ | \ S_{M-3} = S, [v^j]_{M-3}^{M-2} \ ] \label{eqn:D2}   \
\end{align}
Line (\ref{eqn:C2}) may be handled just as in the proof of Lemma \ref{lemma:mm2}. As for line (\ref{eqn:D2}), we use Lemma \ref{lemma:deltaest} and write
\begin{align*}
& \frac{1}{N} \sum_{j=1}^{N} \E^{\Q}[ \delta^N_{M-2}(S_{M-2})\ | \ S_{M-3} = S, [v^j]_{M-3}^{M-2} \ ] \\
& \leq c \ \e + \sum_{i_{M-2}=1}^{m_{M-2}} \beta_{M-2,M-1}^{M-2,i_{M-2},N}\frac{1}{N}\sum_{j=1}^{N}\E^{\Q}[ \psi_{i_{M-2}}(S_{M-2}) \ | \ S_{M-3} = S, [v^j]_{M-3}^{M-2}   ] \
\end{align*}
where $\lim_{N \to \infty} \beta_{M-3,M-2}^{M-2,i_{M-2},N} = 0$ a.s. and we apply the usual separation technique, the details of which we omit. We then obtain
\begin{align*}
&|a^{N}_{M-3}(S) - a_{M-3}(S) |  \\
&\leq c^{N}_{M-3} +  \delta^N_{M-3}(S) \\
& + \sum_{i_{M-1} = 1}^{k_{M-1} } \alpha^{i_{M-1},N}_{M-1,M}\sum_{i_{M-2}}^{k_{M-2}} \alpha^{i_{M-2},i_{M-1},N}_{M-2,M-1} \sum_{i_{M-3}}^{k_{M-3}} \alpha^{i_{M-3},i_{M-2},i_{M-1},N}_{M-3,M-2}\cdot \psi^{M-3,\alpha}_{S,i_{M-3},i_{M-2},i_{M-1} } (S)    \\
& + \sum_{i_{M-2}=1}^{k_{M-2}} \beta^{M-2,i_{M-2},N}_{M-2,M-1} \sum_{i_{M-3}=1 }^{m_{M-3}} \beta^{M-2,i_{M-2},i_{M-3},N}_{M-3,M-2} \psi^{M-2,\beta}_{i_{M-3},i_{M-2}} (S) \
\end{align*}
almost-surely which corresponds to $l = M-2$ and the above random variables satisfy the necessary conditions. The remainder of the induction works exactly analogously to the base case and proofs of the previous Lemmas which establishes the Proposition.
\end{proof}

\section{Overview of Complexity and Multi-Level Monte Carlo/Multi-Grids}
\label{section:complexity}

By viewing our conditional PDEs from the SPDE perspective, we are naturally led to consider Multi-Level Monte Carlo (MLMC)/multi-grid methods as in \cite{Giles} to reduce the complexity. For the direct estimator, our coefficients depend on the expectation of the solutions of our SPDEs, which may be computed using MLMC. Also, the low estimator is expressed as an expectation which may be computed using MLMC. The MLMC adjustment tends to reduce the run times by at least an order of magnitude. MLMC approaches have also appeared in the mixed MC-PDE works of \cite{Ang:2013zl} and \cite{dmd3}.

A pseudo-code description of the methodology based on this idea, and outlined below, is provided in Algorithm \ref{alg:MLMC} in \ref{app:algstate}.

\subsection{Multi-Level Monte Carlo/Multi-Grids}

\subsubsection{MLMC for Computing the Estimated Coefficients}

Assume we have $L$ independent simulations of $v_t$ on $[0,T]$, denoted $\{(v^{(l)}_t)\}_{l=0}^L$, with $N^{v}_l$ paths where $N^{v}_{l} > N^{v}_{l+1}$. Also, let $P_{n}^{(j,l)}(S)$ denote the numerical solution of the conditional PDE on the $j$th path with grid resolutions $N^{\calS}_{l}$, in each dimension, over $[t_n,t_{n+1}]$ where $N^{\calS}_l < N^{\calS}_{l+1}$.

Using the usual multi-level MC appoach, we may write
\begin{equation}
a^{N}_n(S) \approx \left[A^{N_0}_n\right]^{-1} \Bigg[ \frac{1}{N^{v}_0} \sum_{j=1}^{N_0} \phi( v^{j,0}_n ) \cdot  P_n^{(j,0)}(S) + \sum_{l=0}^{L-1} \frac{1}{N_l}\sum_{j=1}^{N_l} \phi( v^{j,l}_n ) \cdot (P^{j,l+1}_n(S) - P^{j,l}_n(S)   ) \Bigg]\,,
\label{eqn:MLMCcoeff}
\end{equation}
where each of $P^{j,l}_n(S)$ is interpolated to have a resolution that matches the grid at level $L$.

A na\"ive implementation may end up carrying out many thousands of interpolations when computing (\ref{eqn:MLMCcoeff}), and also when generating terminal conditions for the PDEs. Our pseudo-code description (Algorithm \ref{alg:MLMC} in \ref{app:algstate}) shows how to keep the total number of interpolations to a minimum.

\subsubsection{MLMC for Low Biased Estimates of the Time Zero Price}

We again carry out $L$ independent simulations of $v_t$ on $[0,T]$, denoted as $\{(v^{(l)}_t)\}_{l=0}^L$, with $N^{v}_l$ paths where $N^{v}_{l} > N^{v}_{l+1}$. Letting $P_0^{(j,l)}(S)$ denote the numerical solution of the conditional PDE on the $j$th path with grid resolution $N^{\calS}_{l}$ over $[0,T]$ where $N^{\calS}_l < N^{\calS}_{l+1}$ we have
\begin{equation*}
V_{l,0}(S) \approx \frac{1}{N^{v}_0} \sum_{j=1}^{N^{v}_0} P_0^{(j,0)}(S) + \sum_{l=0}^{L-1} \frac{1}{N^{v}_l}\sum_{j=1}^{N^{v}_l}  (P_0^{j,l+1}(S) - P_0^{j,l}(S)   )\,,
 \end{equation*}
 where again we interpolate the lower resolution grids to match the highest resolution grid. We do not provide a pseudo-code description for this part as it is relatively simple.

\subsection{An Overview of Complexity}

While a complete complexity analysis lies outside the scope of this paper, we provide an overview of the sources of costs and error, along with a discussion of the complexity of our MLMC-FST scheme.

The direct estimator has three main stages: simulation of paths, solution of conditional PDEs, and the regression step. Assuming an Euler discretization for the simulation of paths, the FST method to solve the conditional PDEs, and Gauss-Jordan elimination for inverting $A^N_{n}$, we obtain the following expression for the algorithm's cost:
\begin{align*}
C &= C_{\text{Path. Sim.}} + C_{\text{PDE} } + C_{\text{Reg.} }  \\
  &\sim N_v \cdot N_t +  N_S^{d_S} \cdot (\log N_S )\cdot N_v + (N_v\cdot d_B + d_B^3 + d_B^2 \cdot N_S^{d_S} ).  \
\end{align*}
The terms in $C_{\text{Reg.}}$ correspond to the costs of constructing $A^N_n$, computing its inverse and multiplying $[A^N_n]^{-1}$ to a length $d_B$ vector at each of the $N_S^{d_S}$ regression sites.

We make the following rough estimate of the algorithm's error at a single point of our grid, $S \in \calS$,
\begin{align*}
\e(S) &\sim \e_{\text{Path. Sim.}} + \e_{\text{PDE} } + \e_{\text{Reg.} }  \\
   &\sim \left(\frac{1}{\sqrt{N_v} } + \frac{1}{N_t} \right) + \left( \frac{1}{N_S^2} + \frac{1}{\sqrt{N_v}} \right) + E(d_B,N_v,N_t,N_S) \
\end{align*}
where determining $E$ is a topic for future papers. For other forms of LSMC, the structure of $\e_{\text{Reg.}}$ has been investigated in \cite{GlassermanYu}, \cite{Stentoft}, \cite{Egloff}, \cite{Belomestny:2011aa}. Another important question for future papers is how the quantities $N_v, N_t, N_S, d_B$ contribute to the error, and ought to be optimally balanced, especially in the context of MLMC that we discuss below.

In Sections \ref{section:NumericalExamples}, alongside implementing the hybrid LSMC/PDE algorithm for Heston-type models, we also investigate how the bias, variance, and cost behave on the various levels of our MLMC-FST scheme in the context of single-period expected values. By the general multi-level theorem (GMLT) (see \cite{Giles2} and \cite{Giles}), we are able to suggest a bound for the complexity of the MLMC-FST component in isolation from the rest of the algorithm. The choice of the highest resolution level for the PDE stage may not be optimal for the entire algorithm as the finest grid is used in the regression stage. The optimal choice of the finest grid should also take the costs and errors of the regression into account; a topic we do not fully address in this paper. On the other hand, the GMLT sheds light on the potential improvements brought about by MLMC to the main bottle-neck of the algorithm. The work of \cite{Giles} also provides an optimal allocation scheme for the number of paths on each level, however, the traditional scheme is for a scalar-valued single-period problem as opposed to a vector-valued multi-period problem. While our MLMC scheme deals primarily with the grids used to solve the PDEs, one could also incorporate multiple discretization levels for the simulated paths, $v_t$. Since the main bottle neck of the algorithm tends to lie with the PDE stage, we leave the paths out of our MLMC scheme for simplicity.

\section{Numerical Examples}
\label{section:NumericalExamples}
In this section, we compare the performance of our algorithm to the standard LSMC algorithm for:
\begin{enumerate}
\item Heston model
\begin{align*}
dS_t = S_t(r \ dt + \sqrt{v_t} \ dW^S_t)\,, \qquad \text{and} \qquad
dv_t = \kappa (\theta - v_t ) \ dt + \eta \sqrt{v_t} \ dW^{v}_t \
\end{align*}
where $d[W^S,W^{v}]_t = \rho \ dt $.
\item Multi-dimensional Heston (mdHeston) model
\begin{align*}
dS^{(1)}_t = S^{(1)}_t(r \ dt + \sqrt{v^{(1)}_t} dW^{(1)}_t)\,,
&&
dv^{(1)}_t = \kappa_1(\theta_1 - v^{(1)}_t) dt + \eta_1 \sqrt{v^{(1)}_t} dW^{(3)}_t\,, \\
dS^{(2)}_t = S^{(2)}_t(r \ dt + \sqrt{v^{(2)}_t} dW^{(2)}_t)\,,
&&
dv^{(2)}_t = \kappa_2(\theta_2 - v^{(2)}_t) dt + \eta_2 \sqrt{v^{(2)}_t} dW^{(4)}_t \
\end{align*}
where $(W^{(i)})_{i=1}^4$ is a 4-dimensional Brownian motion with full correlation structure $\rho = [\rho_{i,j}]$.
\end{enumerate}
We price a Bermudan option with payoff function $h(S)$. In each case, we show results for time-zero prices and optimal exercise boundaries (OEBs). For the standard Heston model, we also implement an explicit finite difference method which we regard as a reference solution as opposed to a benchmark. Finally, we carry out the MLMC-FST tests discussed in Section \ref{section:complexity}.

\subsection{Procedures and Settings for the Main Algorithm and MLMC-FST tests}
We provide a brief procedure for our tests along with Tables \ref{table:basischoice}--\ref{table:LSMCPDEpar}
showing our option and model parameters.

\subsubsection{Main Algorithms}

\label{subsection:procedure}
\begin{enumerate}
\item Carry out $N_{trials}$ of the direct and low estimator for various grid resolutions and numbers of paths. Compute the mean prices and their standard deviations.
\begin{itemize}
\item For the Heston model, report the reference price from the finite difference scheme.
\end{itemize}
\item For each of the $N_{trials}$, using the coefficients from the direct estimator, compute the OEB. By averaging over each grid point, obtain the mean OEB.
\begin{itemize}
\item For the Heston model, report the reference OEB from the finite difference scheme.
\item For the multi-dimensional Heston model, report only certain slices of the average OEB.
\end{itemize}
\end{enumerate}

\subsubsection{MLMC-FST Level Tests}
\begin{enumerate}
\item For each level $l \in \{l_1,\ldots,l_n \}$, compute $\left|\E\left[ P_0^l(S_0) - P_0(S_0) \right]\right|$ over $N_{trials}$. The quantity $P_0(S_0)$ corresponds to the solution of the conditional PDE solved at a relatively high resolution, $N_{S,t}$.
\item For each level $l \in \{l_1,\ldots,l_n \}$, compute $\V \left[ Y_l(S_0) \right] = \E\left[ \left(Y_l(S_0) - \E\left[Y_l(S_0) \right]  \right)^2  \right]$ over $N_{trials}$.
\item For each level $l \in \{l_1,\ldots,l_n \}$, measure the expected CPU time for computing $C_l$, $\E\left[C_l\right]$ over $N_{trials}$.
\end{enumerate}

\begin{table}[h!]
\begin{center}
 \begin{tabular}{|c | c |}
 \hline
 Method & Basis Functions \\ [0.5ex]
 \hline
 LSMC  &    $  (S^{(i)})^m \cdot (S^{(j)})^n \cdot (v^{(k)})^p\cdot (v^{(l)})^q $  for \\
       &   $1 \leq i,j \leq d_S$, $1 \leq k,l \leq d_v$ and $h(S^{(1)},\ldots,S^{(d_S)})^r$, $m,n, p,q,r \in \N_0$  \\
 \hline
 LSMC/PDE   &   $(v^{(i)})^n \cdot (v^{(j)})^m$ for $1\leq i, j \leq d_v$, $n,m \in \N_0$ \\
 \hline
\end{tabular}
\end{center}
\caption{Choices of bases for both methods. We denote by $d_B$ the total number of basis functions used and $deg$ the highest power of the functions in our basis. }
\label{table:basischoice}
\end{table}

\begin{table}[h!]
\begin{center}
 \begin{tabular}{|c | c | c | c | c|}
 \hline
 Model Type & Payoff &   $T$  & $K$ & \text{Exercise Frequency}  \\ [0.5ex]
 \hline
 Heston   & $h(S) = (K-S)_+$  & 1  & $10$ & $T/12$  \\
 \hline
 mdHeston   & $h(S^{(1)},S^{(2)}) = \left(K-\max(S^{(1)},S^{(2)})\right)_+$  & 1  & $10$ & $T/12$  \\
 \hline
\end{tabular}
\end{center}
\caption{Bermudan Option Settings}
\label{table:optnparam}
\end{table}

\begin{table}[h!]
\begin{center}
 \begin{tabular}{|c | c | c | c | c | c | c|}
 \hline
$r$ & $S_0$ & $v_0$  &  $\kappa$ & $\theta$ & $\eta$ & $\rho$  \\ [0.5ex]
 \hline
 0.02 & $10$ & $0.15$  &  5 & 0.16 & 0.9 & $0.1$ \\
 \hline
\end{tabular}
\end{center}
\caption{Parameters used for the Heston model}
\label{table:parmod1}
\end{table}

\begin{table}[h!]
\begin{center}
 \begin{tabular}{| c | c | c | c | c | c | c | c | c | c | c |}
 \hline
 $r$ & $S^{(1)}_0$  & $v^{(1)}_0$ & $\kappa_1$ & $\theta_1$ & $\eta_1$  & $S^{(2)}_0$  & $v^{(2)}_0$ & $\kappa_2$ & $\theta_2$ & $\eta_2$ \\ [0.5ex]
 \hline
 0.025 & 10   & 0.45  & 1.52 & 0.45 & 0.4  & 10   & 0.3  & 1.3 & 0.30 & 0.43   \\
 \hline
\end{tabular}
\[
\begin{matrix}
[\rho_{i,j}]
\end{matrix}
=
\begin{bmatrix}
    1       &   \rho_{S^{1},S^{2}} & \rho_{S^1,v^1}  & \rho_{S^1,v^2} \\
    \rho_{S^2,S^1}  & 1        & \rho_{S^2,v^1}  & \rho_{S^2,v^2} \\
    \rho_{v^1,S^1}  &  \rho_{v^1,S^2} &     1  & \rho_{v^1,v^2} \\
    \rho_{v^2,S^1}  & \rho_{v^2,S^2}   & \rho_{v^2,v^1} &  1 \
\end{bmatrix}
=
\begin{bmatrix}
    1       &   0.2 &   -0.3   & -0.15 \\
    0.2     &     1 &   -0.11  & -0.35 \\
    -0.3    & -0.11 &    1     &  0.2 \\
    -0.15   & -0.35 &    0.2   &  1 \
\end{bmatrix}
\]

\end{center}

\caption{Parameters used for the Multi-Dimensional Heston model }
\label{table:parmod2}
\end{table}

\begin{table}[h!]
\begin{center}
 \begin{tabular}{|c c c|}
 \hline
 $N_{S,l}$  & $N_{S,t}$ &  $N_{trials} $    \\ [0.5ex]
 \hline
 $2^5, \dots , 2^9 $   & $2^{10}$    & $1 \ 000$ \\
 \hline
\end{tabular}
\end{center}
\caption{Parameters used in our MLMC-FST tests  }
\end{table}

\begin{table}[h!]
\begin{center}
 \begin{tabular}{|c c c c c c c |}
 \hline
 $N_{s}$ & $N_{v}$ & $N_{t}$ & $v_{min}$ & $v_{max}$ & $S_{min}$ & $S_{max}$     \\ [0.5ex]
 \hline
 $2^9$ & $2^7$ & $1 \ 000 \ 000$ & 0 & 1& 0 & 53   \\
 \hline
\end{tabular}
\end{center}
\caption{Parameters used for the explicit finite difference method. Boundary conditions are as specified in \cite{Ikonen}. }
\label{table:parfd}
\end{table}

\begin{table}[h!]
\begin{center}
 \begin{tabular}{|c c c|}
 \hline
Model Type  & $N_{S,v}$   & Basis Degree  \\ [0.5ex]
 \hline
 Heston     &  $500 \ 000$  & 3 \\
\hline
mdHeston   &  $500 \ 000$   & 4 \\
\hline
\end{tabular}
\caption{Number of paths and degree of bases used for each model with LSMC. An Euler discretization time step of $N_{step} = 1 \ 000$ was used }
\end{center}
\label{table:parsim}
\end{table}

\begin{table}[h!]
\begin{center}
 \begin{tabular}{|c c c c|}
 \hline
Model Type & $l$  & $N^{(l)}_{v}$     &   $N^{(l)}_{S}$        \\ [0.5ex]
 \hline
Heston/mdHeston & 0     &  $10 \ 000$  &   $2^5$           \\
\hline
Heston/mdHeston & 1     &  $1 \ 000$   &   $2^6$           \\
\hline
Heston/mdHeston & 2     &  $100   $   &   $2^7, 2^8, 2^9$  \\
\hline
\end{tabular}
 \begin{tabular}{|c c|}
 \hline
Model Type & Basis Degree        \\ [0.5ex]
 \hline
Heston     &  3    \\
\hline
mdHeston   &  3     \\
\hline
\end{tabular}

\begin{center}
 \begin{tabular}{| c c |}
 \hline
 $(\log S^{(i)}/S^{(i)}_0)_{min}$  & $(\log S^{(i)}/S^{(i)}_0)_{max}$ \\ [0.5ex]
 \hline
 $-3$  & $3$  \\
 \hline
\end{tabular}
\end{center}
\caption{MLMC, FST, and basis set up for LSMC-PDE algorithm. Multiple values of $N^{(2)}_S$ correspond to the  different trials that will be presented. These settings are applied to both the Heston and multi-dimensional Heston models. \label{table:LSMCPDEpar}}
\end{center}
\end{table}

\begin{table}[h!]
\begin{center}
 \begin{tabular}{| c | c | c | c | c | c |}
 \hline
 $N^{(2)}_{s}$   & Est. Type &   $0.95 \cdot  S_0$ & $ S_0$ & $1.05\cdot S_0 $      & Run Time (s)    \\ [0.5ex]
 \hline
 $2^{7}$  &  Direct       &  1.6747 (0.0011)   &  1.4541 (0.0012)   & 1.2598 (0.0013)  &  7     \\
 		  &  Low          & 1.6746 (0.0013)    &  1.4540 (0.0015) & 1.2597 (0.0016)  &  6    \\
\hline
 $2^{8}$  &  Direct       &  1.6739 (0.0013)   & 1.4534 (0.0014)  & 1.2591 (0.0015)      &  7 \\
 		  &  Low          &  1.6738 (0.0012)   & 1.4532 (0.0013)  &  1.2588 (0.0014)      &   6  \\
\hline
 $2^{9}$  &  Direct       &  1.6735 (0.0011)   & 1.4530 (0.0013)  & 1.2586 (0.0014)   &   7   \\
 		  &  Low          &  1.6735 (0.0014)   & 1.4529 (0.0016)  & 1.2585 (0.0017)   &   6  \\
 \hline
\end{tabular}
\end{center}
\caption{Resulting ATM and non-ATM price estimates for the Heston model LSMC-PDE. We show results for $(S_0,v_0)$ for various $S_0 \in \calS$ and a fixed $v_0$. Reference prices for $0.95 \cdot S_0, \  S_0, \ 1.05\cdot S_0$ are 1.6712, 1.4507, 1.2565, respectively.  The value in brackets correspond to standard deviations.}
\label{table:resultslsmcpde}
\end{table}

\begin{table}[h!]
\begin{center}
 \begin{tabular}{|c |c |c |c|}
 \hline
 Estimate Type  & $(S_0,v_0)$             & Run Time (s)  \\ [0.5ex]
 \hline
 Direct         &  1.4494 (0.0020)  &     47      \\
 Low            &  1.4487 (0.0023)  &     38      \\
 \hline
\end{tabular}
\end{center}
\caption{Resulting ATM price estimate for the Heston model using LSMC algorithm. The reference value obtained from finite difference is 1.4507.  The value in brackets correspond to standard deviations. }
\label{table:resultslsmc}
\end{table}

\begin{table}[h!]
\begin{center}
 \begin{tabular}{||c c c||}
 \hline
 Slope Type & Estimate     & $R^2$   \\ [0.5ex]
 \hline\hline
 $\alpha$      &  2.12    & 0.9993        \\
 $\beta$       &  4.16    & 0.9989         \\
 $\gamma$      &  0.30    & 0.9409       \\
 \hline
\end{tabular}
\end{center}
\caption{Resulting values for MLMC-FST tests for the Heston model. The parameters $\alpha, \beta, \gamma$ correspond to slopes of logarithmic bias, variance, and cost.  }
\label{table:fst1}
\end{table}

\begin{table}[h!]
\begin{center}
 \begin{tabular}{|c | c | c | c|}
 \hline
 $N^{(2)}_{s}$ &  Est Type &   $ ( S^{(1)}_0, S^{(2)}_0 ) $   & Run Times (s)  \\ [0.5ex]
 \hline
$2^{7}$        &  Direct       &  1.1853 (0.0055)        &     67              \\
		       &  Low          &  1.1852 (0.0049)        &     43              \\
 \hline
$2^{8}$        &  Direct       &  1.1834 (0.0055)        &     102              \\
	        	   &  Low          &  1.1833  (0.0049)       &      64               \\
\hline
$2^{9}$        &  Direct       &  1.1836 (0.0057)        &       177            \\
		       &  Low          &  1.1830 (0.0060)        &     105              \\
 \hline
\end{tabular}
\end{center}

\begin{center}
 \begin{tabular}{|c | c | c | c|c | c|}
 \hline
 $N^{(2)}_{s}$ &  Est Type &   $ ( 0.95\cdot S^{(1)}_0, S^{(2)}_0 ) $ &   $ ( 1.05\cdot S^{(1)}_0, S^{(2)}_0 ) $  &   $ ( S^{(1)}_0, 0.95 \cdot S^{(2)}_0 ) $ & $ (S^{(1)}_0, 1.05 \cdot S^{(2)}_0 ) $   \\ [0.5ex]
 \hline
$2^{7}$        &  Direct       &  1.2526 (0.0056)    &  1.1209 (0.0054)     & 1.2846 (0.0056)  & 1.0931 (0.0054)         \\
		       &  Low          &  1.2525 (0.0050)    &  1.1208 (0.0049)     & 1.2845 (0.0050)  & 1.0940 (0.0048)    \\
\hline
$2^{8}$        &  Direct       &  1.2506 (0.0055)    &   1.1191 (0.0054)    & 1.2826 (0.0056)  & 1.0914 (0.0054)     \\
	        	   &  Low          &    1.2506 (0.0050)  &   1.1190 (0.0049)    & 1.2825 (0.0050)  & 1.0913 (0.0048)   \\
\hline
$2^{9}$        &  Direct       &  1.2508 (0.0058)    & 1.1194 (0.0057)      & 1.2828 (0.0058)  &  1.0916 (0.0056)   \\
		       &  Low          &  1.2502 (0.0061)    & 1.1187 (0.0059)      & 1.2821 (0.0061)  &  1.0910  (0.0059)   \\
 \hline
\end{tabular}
\end{center}
\caption{Resulting ATM and non-ATM price estimates for LSMC-PDE. We show results for $(S_0,v_0)$ for various $S_0 \in \calS$ and a fixed $v_0$.  The value in brackets correspond to standard deviations. }
\label{table:resultslsmcpdeprml2d}
\end{table}

\begin{table}[h!]
\begin{center}
 \begin{tabular}{|c| c |c |c|}
 \hline
 Estimate Type & $(S_0,v_0)$   &  Run Times (s)              \\ [0.5ex]
 \hline
 Direct         &  1.2121   (0.0016)    & 83        \\
 Low            &  1.1765   (0.0018)    & 73      \\
 \hline
\end{tabular}
\end{center}
\caption{Resulting ATM price statistics for the LSMC algorithm. The value in brackets correspond to standard deviations.}
\label{table:resultslsmc2d}
\end{table}

\begin{table}[h!]
\begin{center}
 \begin{tabular}{|c c c|}
 \hline
 Slope Type & Estimate     & $R^2$   \\ [0.5ex]
 \hline
 $\alpha$      &  2.03    & 0.9995        \\
 $\beta$       &  2.63    & 0.9640         \\
 $\gamma$      &  1.65    & 0.9785       \\
 \hline
\end{tabular}
\end{center}
\caption{Resulting values for MLMC-FST tests for the multi-dimensional Heston model. The parameters $\alpha, \beta, \gamma$ correspond to slopes of logarithmic bias, variance, and cost.  }
\label{table:fst2}
\end{table}

\clearpage

\subsection{Discussion of Results}

We provide a discussion of the results from the two examples we investigated above. We first provide example-specific comments for each model, followed by comments which apply to both examples.

\subsubsection{Heston model}

Our choice of parameters for this problem are borrwed from \cite{Ikonen} with a lower value for $r$, the risk free rate, and higher maturity. We also differ in that we value a Bermudan option instead of an American.

Based on Figure \ref{fig:lsmcpdeprmlbdy}, we see the hybrid algorithm's errors appear along the interface of the holding and exercise regions and it is able to closely approximate the true boundary across all exercise dates. The consistency arises from the comparison of the exercise value and continuation value being broken into a collection of cross-sectional comparisons for each $S \in \calS$ as discussed in Section \ref{section:LSMC-PDEalg}. At time $n$, for each $S' \in \calS$, we have an estimate of the function $C^N_n(S',v)$ which is monotonically increasing in $v$. Our goal is then to locate the points $v^* \in \R$ such that $C^N(S',v^*) = h(S')$ which is simple to compute since $h(S')$ is a deterministic constant. We do, however, see some noise in determining $v^*$ across trials which is due to randomness in $C^N(S',\cdot)$. From Figures \ref{fig:tvrbdy}, we see that the LSMC boundary is not particularly accurate for all exercise dates, with greater accuracy at the later dates. This is due to the accumulation of errors as one moves recursively through time and the lack of paths that are in-the-money at earlier exercise dates.

As indicated in Tables \ref{table:resultslsmcpde} and \ref{table:resultslsmc} the hybrid algorithm is better able to estimate time-zero prices, in comparison to the reference value. This accuracy in our prices follows from the quality of our optimal exercise boundaries (OEBs). We note that the direct estimator and low estimator are fairly close to each other and suggest that the low estimator is somewhat redundant in this case. The full LSMC algorithm also provides a direct estimate of the price that is also close to the low estimator, however, this price is biased too low as evidenced by the reference value. Since the low estimate of the LSMC-PDE algorithm is higher than the low estimate of the LSMC algorithm, by definition it is better. We also note the hybrid algorithm provides estimates for other $S_0 \in \calS$ and fixed $v_0$, as opposed to a single point in LSMC. The non-ATM price estimates are of the same quality as the ATM estimate.

For run-times, we see that regardless of the choice of $N^{(2)}_{\calS}$, the run times remain effectively unchanged which implies the $N^{(2)}_v = 100$ conditional PDEs solved at $N^{(2)}_{\calS}$ are highly out-weighed by the work done at resolutions $N^{(0)}_{\calS}, N^{(1)}_{\calS}$. We also see the direct estimator takes an extra second compared to the low estimator, roughly implying that the multi-grid PDE stage takes approximately $86\%$ of the run time.

Lastly, there is some discrepancy in the prices obtained by finite differences and the LSMC-PDE algorithms, mostly in the third decimal place. This discrepancy is likely due to the bias introduced from path discretization, interpolation of time-zero prices, and spatial discretization.

It is worth noting that the PDE approach is best for this $2d$ example as it has lower complexity. This example is used to demonstrate the performance of the LSMC-PDE algorithm in a fully observable setting.

\subsubsection{Multi-dimensional Heston model}

As our OEBs are four-dimensional objects, they are not fully observable; hence we show $v$-slices of the OEB, as we have some intuition about how they should look. In Figures \ref{fig:slice1}, \ref{fig:slice2}, \ref{fig:slice3}, we show three types of slices $v_1 = \theta_1, v_2 = \theta_2$, the `central slice', along with two other off-center slices. For brevity we show only the last, middle, and first exercise dates.

Overall the hybrid and LSMC algorithm have comparable precision in determining the region over 100 trials. However, we note that the boundaries produced by the LSMC-PDE algorithm are considerably different, just as in the standard Heston model. In looking at the price estimates  in Tables \ref{table:resultslsmcpdeprml2d}, \ref{table:resultslsmc2d}, we see that the lower estimates produced by the hybrid algorithm are generally lower, and so we are more inclined to trust the OEBs produced by the hybrid algorithm.

In terms of price estimates, the tightness of the direct and low estimator for LSMC-PDE suggest that, again, the low estimator is somewhat redundant. In this case, for LSMC, the direct estimator is considerably higher than the low estimator; indicating a relatively high degree of bias introduced by the regression scheme. We also note that, again for the hybrid algorithm, we are able to get accurate non-ATM price estimates. This suggests LSMC boundary is consistently wrong in many places.

Looking at run times, we see that unlike in the standard Heston model, the PDEs solved at resolution $N^{(2)}_{\calS}$ weigh significantly on the run times. Also, for the direct estimator, we see that the PDE stage typically takes approximately $63 \%$ of the run time, indicating an increase in regression costs compared to the standard Heston model.

\subsubsection{Overall Comments}
%
%We note that we encounter different levels of variance through the time-zero prices and OEBs.

Our choice of MLMC scheme was for simplicity and conservativeness, and hence not optimized, as alluded to before.

One issue that arises from our results is the question of whether the hybrid algorithm is more efficient. For the Heston model, we see that the hybrid algorithm provides substantial variance reduction in time zero prices and OEBs especially given that LSMC was allowed a larger computational budget. It is interesting to note, however, that while the standard deviations for prices are relatively close to each other, the OEBs for LSMC ``seem" much noisier than the hybrid algorithm's, although this is based on subjective examination. For the multi-dimensional Heston model, we note that the variance of the hybrid algorithm's prices, for a comparable computational budget, is higher, although their boundaries have very similar levels of noisiness. This again points to a disconnect in the variance of the two estimated objects. On the other hand, it is important to note that for the hybrid algorithm, we obtain prices for all $S \in \calS$ and not simply the ATM point, as in standard LSMC. Having the prices for these $S \in \calS$ also allows us to compute the $\Delta$ and $\Gamma$ of the option without any extra simulations. As a result, the hybrid algorithm provides considerably more information than standard LSMC. Also, the information provided by the hybrid algorithm tends to be more accurate in terms of quality of the boundary and bias, as noted before.

Next, we turn to our MLMC-FST tests. From Table \ref{table:fst1} and \ref{table:fst2} and Figures \ref{fig:1p1B} to \ref{fig:2p2V}, by the GMLT it is suggested, since $\alpha > \frac{1}{2} \min(\beta,\gamma)$, and $\beta > \gamma$, that optimal complexity is attainable. That is, there exists a collection of levels and a path allocation such that one can obtain a RMSE of $O(\e)$ with cost $O(\e^{-2})$ (up to a certain time-step resolution). This, again, is for our single-period computations in isolation from the entire hybrid algorithm. As one moves from the $2d$ example to the $4d$ example, we see that the accuracy across levels, measured by $\alpha$, remains unchanged. On the other hand, the variances and costs increase considerably as we increased the number of volatility processes and dimension of $S$. These results suggest that we may not conclude optimal complexity, via the GMLT, in the case when $d_S = 3$,  $d_v = 3$ analogue, although we leave this investigation to future study.

All run times were measured on a 3.60 GHz Intel Xeon CPU using  \textsc{Matlab} 2016. For the results in Tables \ref{table:resultslsmcpde} to \ref{table:resultslsmc2d} all run-times were measured by repeating the following process three times: clearing the memory, running the algorithm, and measuring the run time. The reported run time was then taken to be the maximum of the three.

\section{Conclusions}

By combining conditional PDEs with the LSMC approach of \cite{TV}, we have developed an algorithm that improves traditional LSMC methods in the context of SV models for medium dimensional problems. The hybrid algorithm also provides an extension of PDE methods beyond lower dimensional problems, at the expense of variance and a semi-global solution.

From a theoretical perspective we provide a proof of almost-sure convergence which uses geometric arguments and requires a few modifications to the basis, regression matrix, and payoff function. The proof is also highly dependent on separability of the model. Future research into theoretical properties of the algorithm ought to search for a proof with more natural assumptions and farther reach in terms of types of SV models. Another important point is that the proof is for a stylized form of the algorithm where $\calS$ is a continuum rather than a finite grid of points. It seems a more appropriate proof would make assumptions on the properties of the conditional expectations, viewing them as numerical solutions of the conditional PDE. As discussed in Section \ref{section:complexity}, a major question is quantifying the rate of convergence which essentially requires a Central Limit Theorem type result as one component. To obtain this result, it would appear that our separation technique will not carry over, and so one would likely again require a more sophisticated treatment of the conditional expectation functions. Our treatment in Section \ref{section:complexity} essentially assumes a CLT holds. With a quantification of the error associated to the regression scheme, one can then quantify the entire algorithm's complexity; a topic for future work. Also, as alluded to before, there is a need to find an optimal allocation scheme for the MLMC component of the algorithm and prove convergence under this scheme.

There is also the question of developing duality-based high biased estimators as in \cite{duality2} and \cite{duality1}. Such estimates would allow us to form proper confidence intervals for our time-zero prices.

In terms of computing sensitivities, as mentioned before, the details may be found in \cite{Farahany}; one can also find applications to the double Heston model (\cite{dHeston}), a mean-reverting commodity model with jumps, and concrete derivations of conditional PDEs.

\clearpage

\section{Bibliography}

\bibliographystyle{chicago}
\bibliography{lsmc-pde-biblio}
\nocite{*}

%\newpage

\appendix

\section{Formal Algorithm Descriptions}
\label{app:algstate}

\subsection{LSMC-PDE Algorithm}

\begin{algorithm*}[h!]
  \caption{LSMC-PDE: Exercise Boundaries and Direct Estimator}\label{alg:desc}
  \begin{algorithmic}[1]

      \State Simulate $N$ paths of $v_t$

      \For{$n = M:2$}

      \If{$n == M$}
        \State Set $V_n(s_i,v) = h_n(s_i)$
        \Comment{For $i = 1,...,N^{d_S}_s$}

      \EndIf  \Comment{Initialize boundary for PDE solver}

        \For{$j = 1:N$} \label{forloop}
            \State Compute $\E\left[ \left.V_n(S_{n},v^j_{n})\right. \mid  \left.S_{n-1} = s_i, [v^j]_{n-1}^{n}\right. \right]$ \Comment{For $i=1,...,N^{d_S}_s$}

        \EndFor

        \For{$i=1:N^{d_S}_s$} \label{regstep}
          \State Regress $\{ \ e^{-r\Delta t}\E \left[ \left.V_n(S_{n},v^j_{n})\right. \mid  \left.S_{n-1} = s_i, [v^j]_{n-1}^{n}\right.  \right] \ \}_{j=1}^{N}$ onto $\{\phi_m(v)\}_{m=1}^{d_B} $

          \Comment{Obtain $[ a_{n-1}(s_i)]_{i=1}^{N^{d_S}_s} $ }

          \State Set $C_{n-1}(s_i,v) = a_{n-1}(s_i)\cdot \phi(v)$

        \EndFor \Comment{Obtain a matrix of dimension $N^{d_S}_s \times d_B$}

        \State Set $V_{n-1}(s_i,v) = \max(h_{n-1}(s_i), C_{n-1}(s_i,v) )$ \Comment{For $i=1,...,N^{d_S}_s$}

      \EndFor

      \For{$j=1:N$}

      \State Compute $e^{-r\Delta t}\E\left[ \left.V_1(S_{1},v^j_1)\right. \mid \left. S_{0} = s_i, [v^j]_{0}^{1}\right. \right]$
      \Comment{For $i=1,...,N_s^{d_S}$}
      \EndFor

      \State Set $V_{h,0}(s_i,v_0) = \max\left( h_{0}(s_i),\frac{1}{N}\sum_{i=1}^{N}e^{-r\Delta t}\E \left[ \left.V_1(S_{1},v^i_1)\right. \mid   \left.S_{0} = s_i, [v^i]_{0}^{1}\right. \right] \right)$

      \Comment{For $i=1,...,N_s^{d_S}$}

      \Comment{High estimate of the time-zero prices}
      \State \textbf{return} $V_{h,0}(s_i,v_0)$ for $i=1,...,N^{d_S}_s$ %and $V_{l,0}(s',v_0)$

    %\EndProcedure
  \end{algorithmic}
\end{algorithm*}

\begin{algorithm*}[h!]
  \caption{LSMC-PDE: Lower Estimator}\label{alg:desc2}
  \begin{algorithmic}[1]

      \State Simulate $N_{v,2}$ paths of $v_t$
      \For{$j=1,...,N_{v,2} $}
      		\For{$n=M:2$}
      			\If{$n == M$}
        			\State Set $V^j_n(s_i,v^j_{M}) = h_n(s_i)$
        			\Comment{For $i=1,...,N_s^{d_S}$}
        		\EndIf

        		\State Compute $U^j_{n-1}(s_i) = e^{-r\Delta t} \E[V_n(S_n,v^j_n) \ | \ S_{n-1} = s_i, [v^i]_{n-1}^n  ]$
        		\State Set $V^j_{n-1}(s_i) = U^j_{n-1}(s_i) I_{\Gamma_n}(s_i,v^j_{n-1}) + h(s_i)I_{\Gamma^c_n}(s_i,v^j_{n-1})  $
        		\Comment{For $i=1,...,N_s^{d_S}$}
        	\EndFor

        	\State Set $V^j_{l,0}(s_i) = e^{-r \Delta t} \E[V_1(S_1,v^j_1) \ | \ S_{0} = s_i, [v^i]_{n-1}^n  ]$
        	\Comment{For $i=1,...,N_s^{d_S}$}

      \EndFor
      \State Set $V_{l,0}(s_i) = \frac{1}{N_{v,2}} \sum_{j=1}^{N_{v,2}} V^j_{l,0}(s_i)$
       \Comment{For $i=1,...,N_s^{d_S}$}
    \State \textbf{return} $V_{l,0}(s_i,v_0)$ for $i=1,...,N^{d_S}_s$ %and $V_{l,0}(s',v_0)$
    \end{algorithmic}
\end{algorithm*}

\subsection{MLMC/Multi-Grids}
We provide guidelines for computing $a_n(S)$ over the time interval $[t_n,t_{n+1}]$ in a manner that keeps the number of interpolations to a minimum. Given our coefficient matrix $a_{n+1}(S)$ at level $L$, we carry out the following.

\begin{algorithm*}[h!]
  \caption{MLMC/Multi-Grids For $a_n(S)$}\label{alg:MLMC}
  \begin{algorithmic}[1]
      \For{$l = 0:L-1$}
      		\State Interpolate $a_{n+1}(S)$ from resolution level $L$ to level $l$.		
      \EndFor
      \For{$l = L-1:0$}
      		\For{$j=1:N^{v}_l$}
      			\State Generate boundary conditions for the PDEs to be solved at resolutions $l$ and $l+1$.

      			\Comment Requires computing the matrix $a_n(S)\cdot \phi(v_n^{j})$ using $a_n(S)$ at grid level $l, \ l+1$.

      			\State Compute $P^{j,l}_n(S)$ and $P^{j,l+1}_n(S)$ via a numerical PDE technique.

      		\EndFor

      		\State Compute $\frac{1}{N^{v}_l}\sum_{j=1}^{N^{v}_l} \phi(v^{j,l}_n) P^{j,l}_n(S) $ and $\frac{1}{N^{v}_l}\sum_{j=1}^{N^{v}_l} \phi(v^{j,l}_n) P^{j,l+1}_n(S) $
      		\State Interpolate both to grid level $L$
      		\State Compute $\frac{1}{N^{v}_l}\sum_{j=1}^{N^{v}_l} \phi(v^{j,l}_n) \Big( P^{j,l+1}_n(S) - P^{j,l}_n(S) \Big)$
      \EndFor

      Repeat the above procedure for level $0$ and compute $a_n(S)$ which is defined on grid level $L$ as in (\ref{eqn:MLMCcoeff}).

  \end{algorithmic}
\end{algorithm*}

\pagebreak

\pagebreak

\section{Optimal Exercise Boundaries}
\label{section:OEB}
\vspace{5cm}
\begin{figure}[H]
\centering
	\includegraphics[width = 0.85\textwidth]{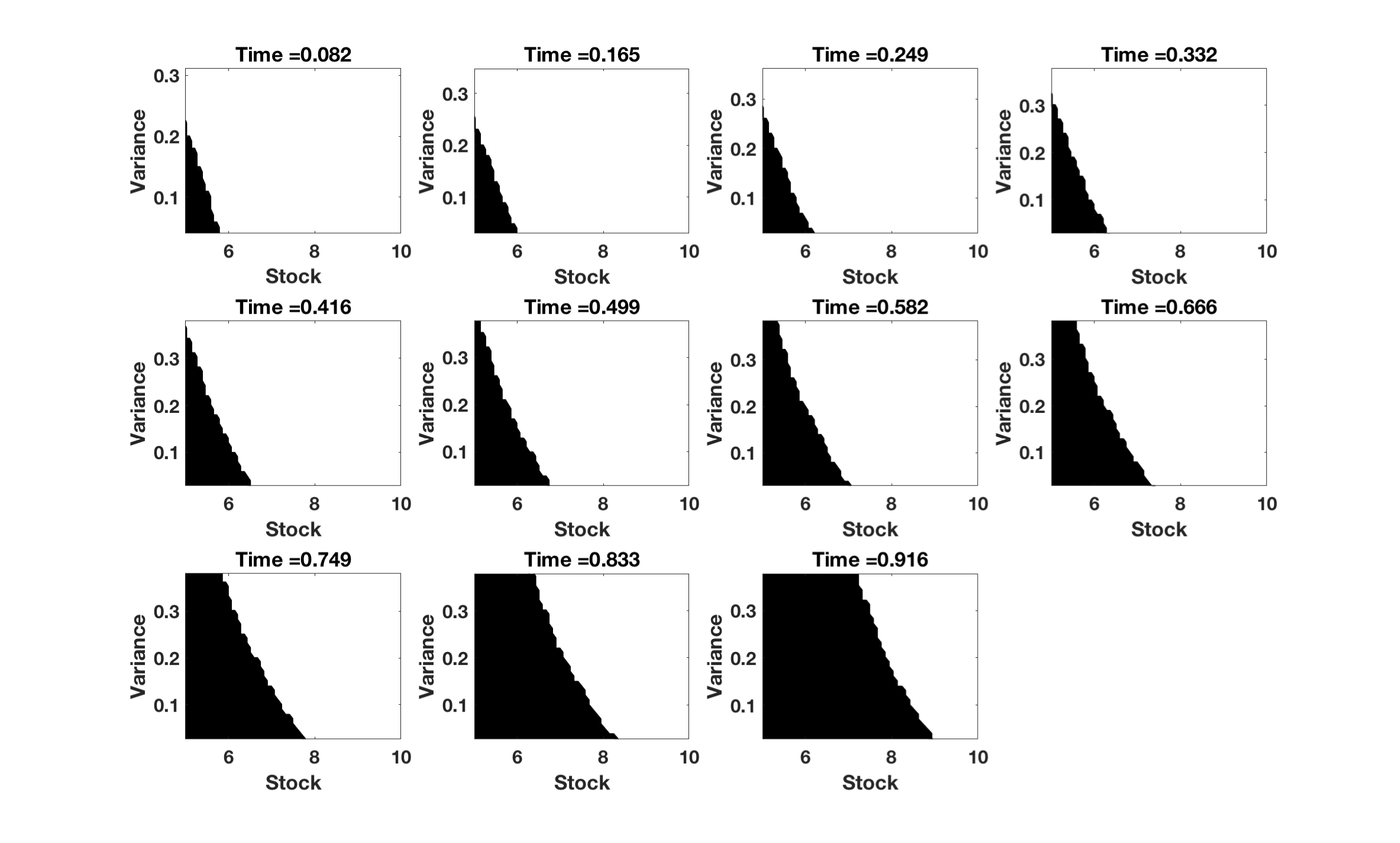}
  \caption{Reference Boundary Obtained by Finite Difference. The black regions indicate the exercise areas, and the white regions indicate the holding areas. }
  \label{fig:fdbdy}
\end{figure}

\pagebreak
\begin{figure}[H]
%\vspace{5cm}
\centering
	\includegraphics[width = 0.85\textwidth]{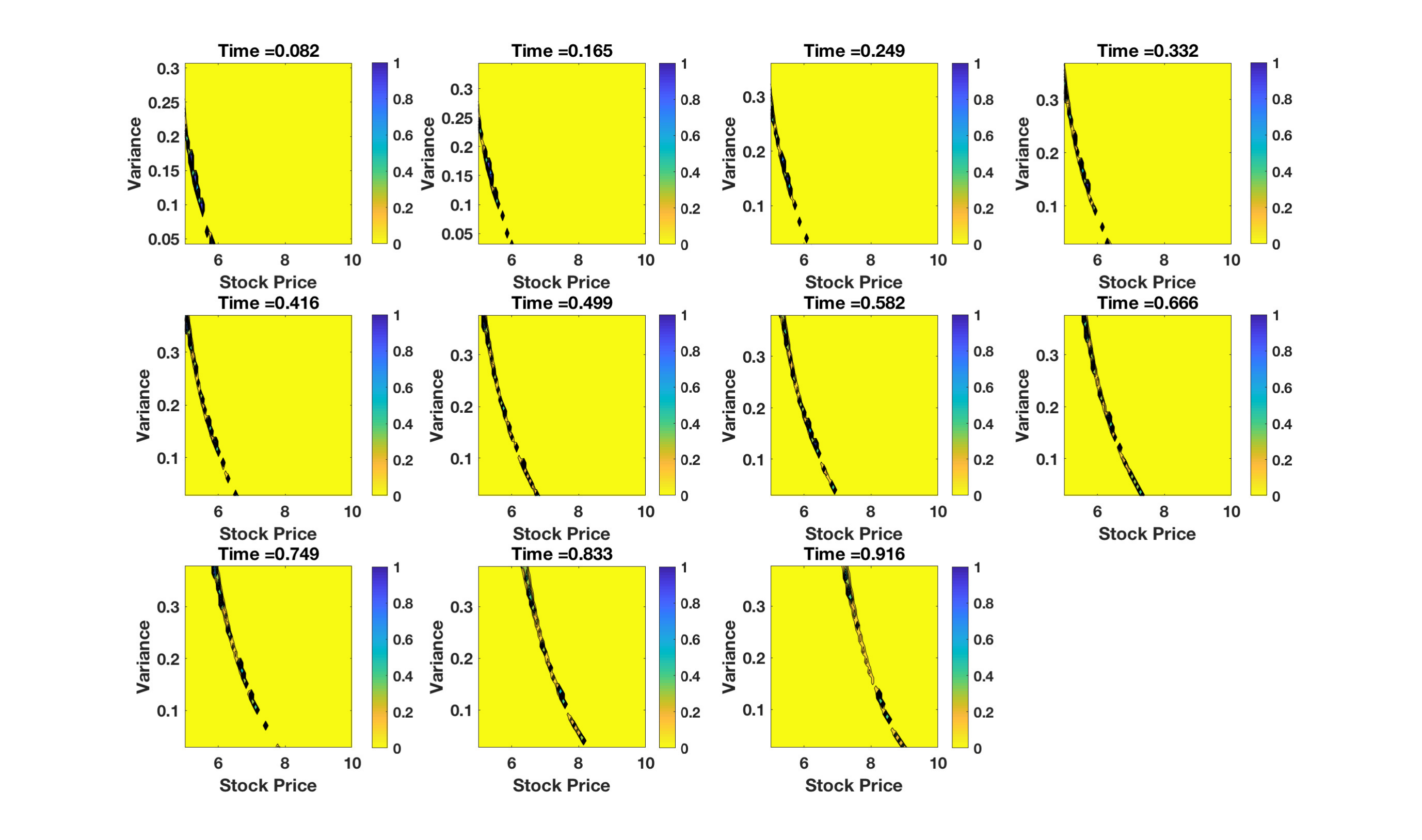}
  \caption{Difference boundary for the LSMC-PDE algorithm corresponding to $N^{(2)}_S = 2^8$. Yellow regions indicate correctness with probability 1, blue regions indicate incorrectness with probability 1. }
  \label{fig:lsmcpdeprmlbdy}
\end{figure}

%\pagebreak
\begin{figure}[H]
%\vspace{5cm}
\centering
	\includegraphics[width = 0.85\textwidth]{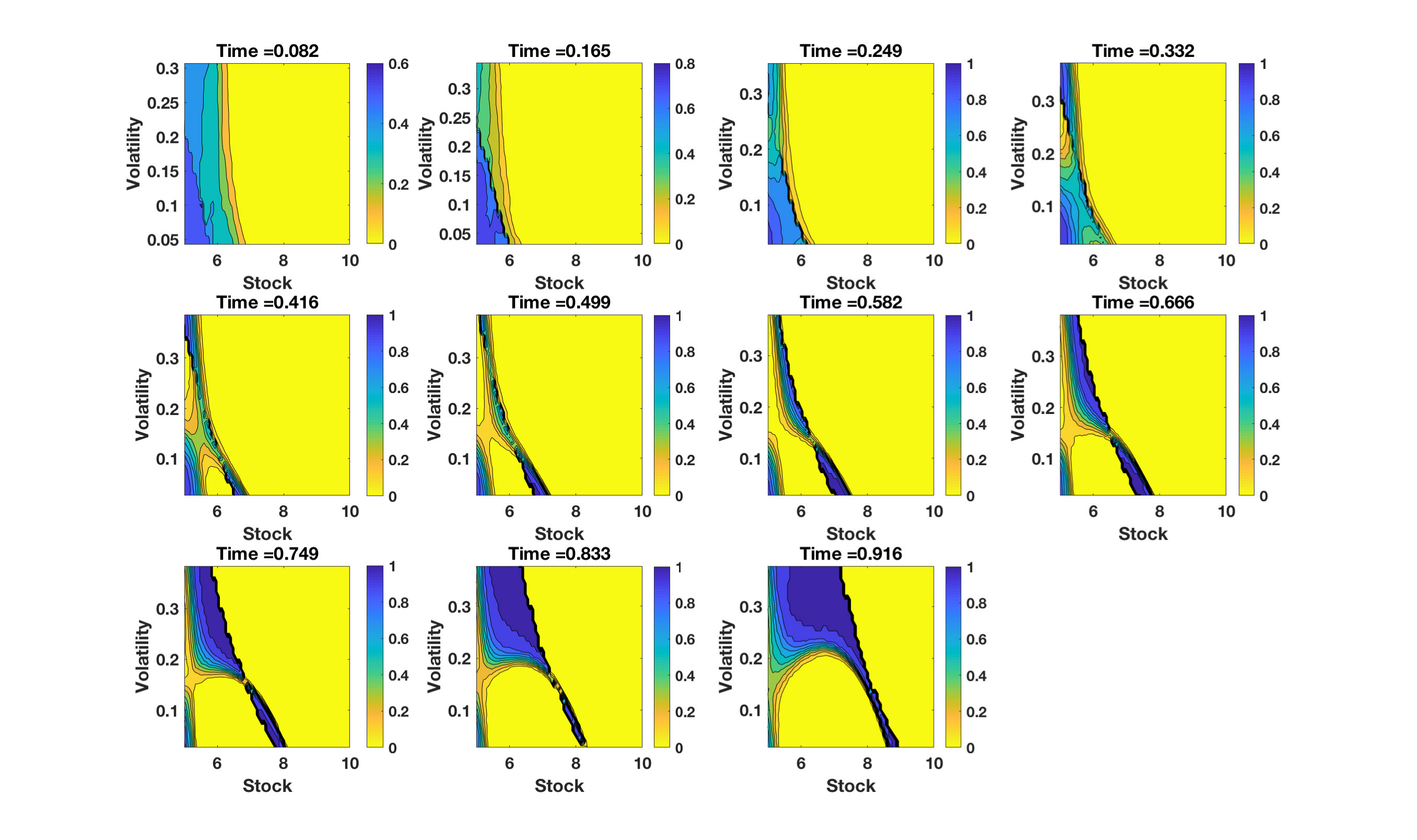}
  \caption{Difference boundary for the LSMC algorithm. Yellow regions indicate correctness with probability 1, blue regions indicate incorrectness with probability 1. }
  \label{fig:tvrbdy}
\end{figure}

\newpage

\begin{figure}[H]
\centering
\begin{minipage}{0.49\textwidth}
	\includegraphics[width=\textwidth]{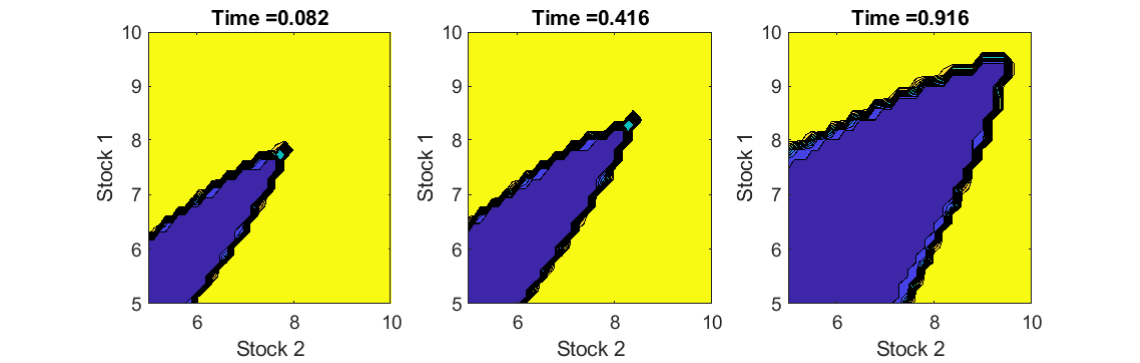}
	\includegraphics[width=\textwidth]{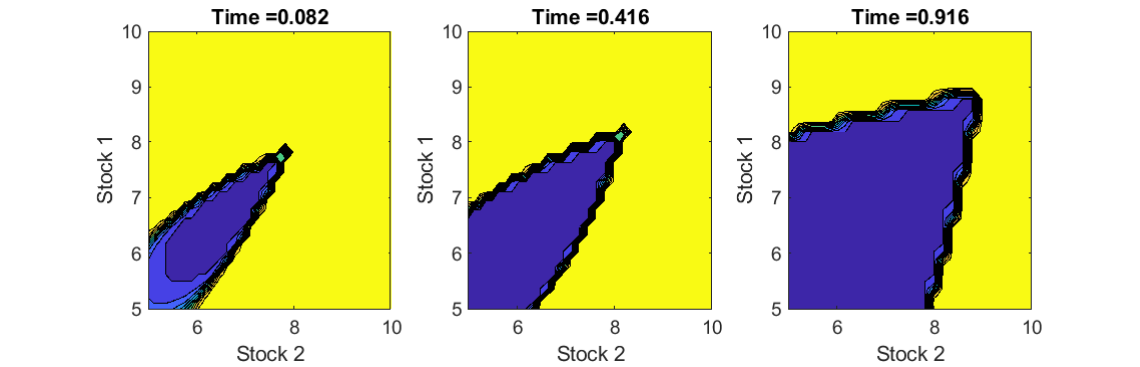}
\end{minipage}
\begin{minipage}{0.49\textwidth}
  \caption{A slice of the boundaries when $v_1 = \theta_1, v_2 =  \theta_2$. Yellow regions indicate holding with probability 1, blue regions indicate exercising with probability 1. The top and bottom figures correspond to LSMC-PDE and LSMC, respectively. The resolution for LSMC-PDE is $N_{\cal{S} }^{(2)} = 2^8$.}
  \label{fig:slice1}
\end{minipage}
\end{figure}

\begin{figure}[H]
\centering

\begin{minipage}{0.49\textwidth}
	\includegraphics[width=\textwidth]{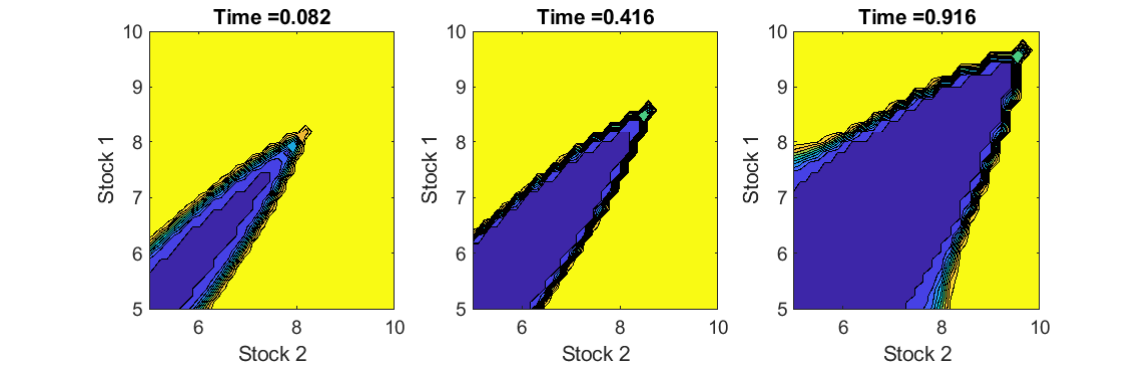}
	\includegraphics[width=\textwidth]{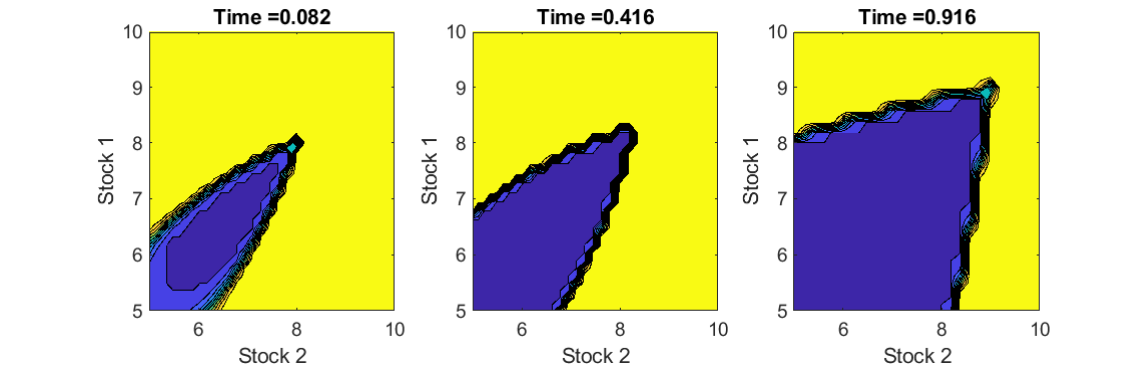}
\end{minipage}
\begin{minipage}{0.49\textwidth}
  \caption{A slice of the boundaries when $v_1 = 0.5 \cdot \theta_1, v_2 =  \theta_2$. Yellow regions indicate holding with probability 1, blue regions indicate exercising with probability 1. The top and bottom figures correspond to LSMC-PDE and LSMC, respectively. The resolution for LSMC-PDE is $N_{\cal{S} }^{(2)} = 2^8$. }
  \label{fig:slice2}
\end{minipage}
\end{figure}

\begin{figure}[H]
%\vspace{5cm}
\centering
\begin{minipage}{0.49\textwidth}
	\includegraphics[width=\textwidth]{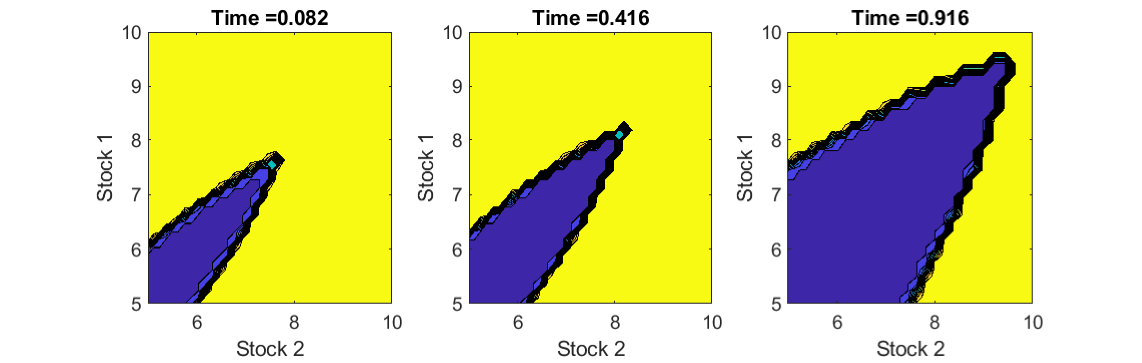}
	\includegraphics[width=\textwidth]{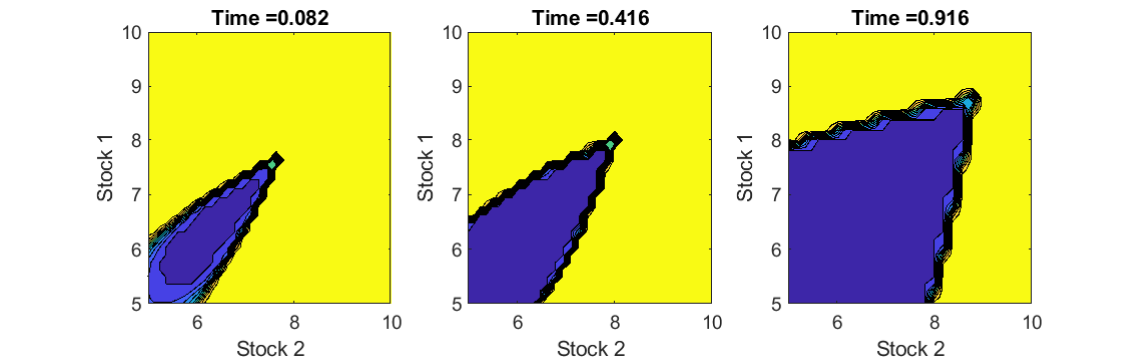}
\end{minipage}
\begin{minipage}{0.49\textwidth}
  \caption{A slice of the boundaries when $v_1 =  \theta_1, v_2 = 1.5 \cdot \theta_2$. Yellow regions indicate holding with probability 1, blue regions indicate exercising with probability 1. The top and bottom figures correspond to LSMC-PDE and LSMC, respectively. The resolution for LSMC-PDE is $N_{\cal{S} }^{(2)} = 2^8$.}
  \label{fig:slice3}
\end{minipage}
\end{figure}

\newpage

\section{MLMC-FST Test Plots}

\begin{figure}[!htb]
\minipage{0.32\textwidth}
  \includegraphics[width=\linewidth]{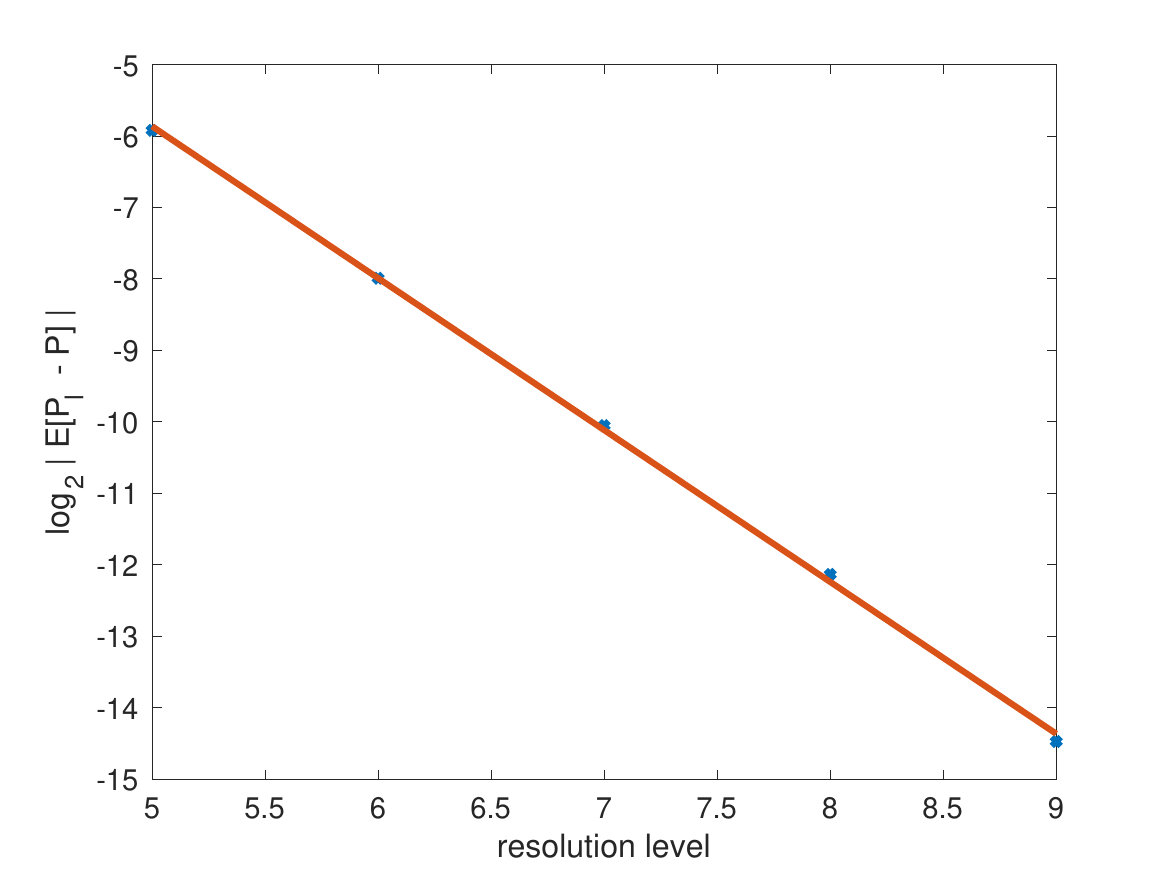}
  \caption{Bias on each resolution level. The slope of the plot is denoted as $\alpha$. }\label{fig:1p1B}
\endminipage\hfill
\minipage{0.32\textwidth}
  \includegraphics[width=\linewidth]{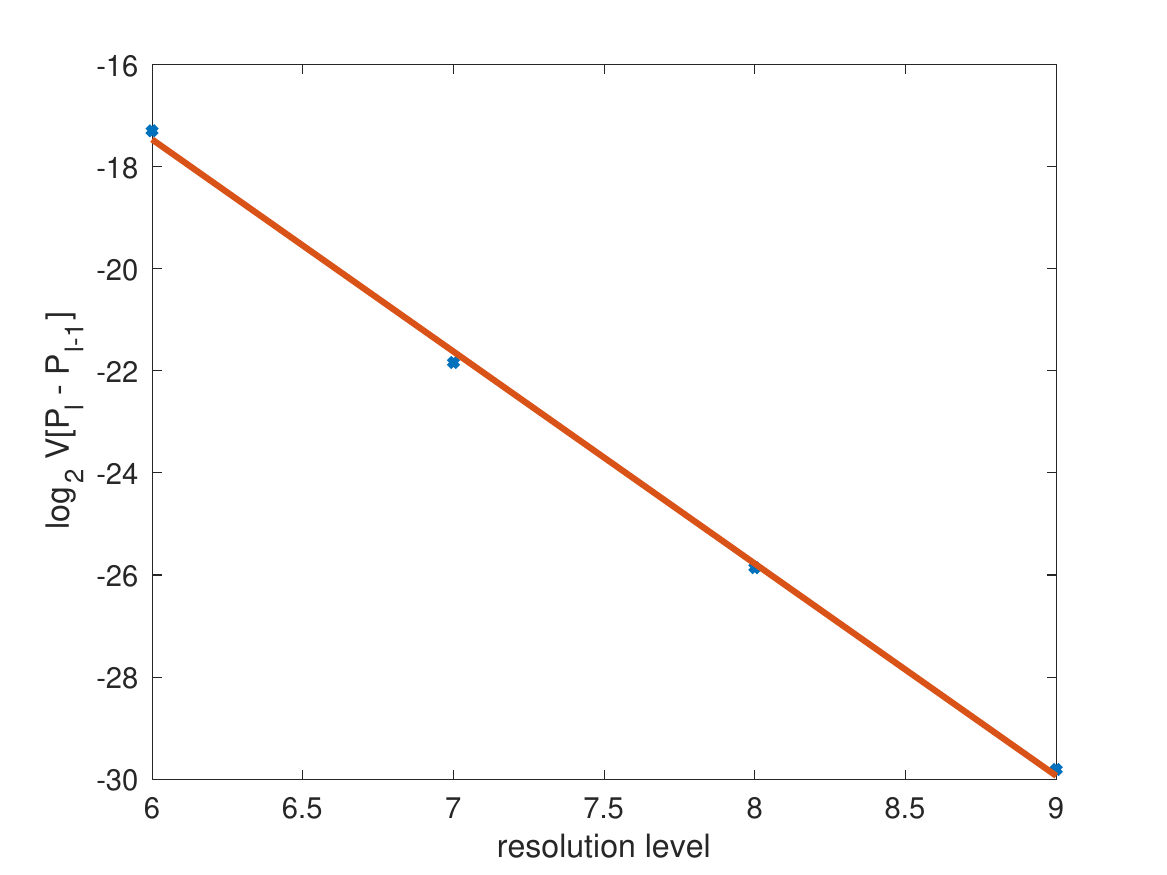}
  \caption{Variance on each resolution level. The slope of the plot is denoted as $\beta$.}\label{fig:1p1V}
\endminipage\hfill
\minipage{0.32\textwidth}%
  \includegraphics[width=\linewidth]{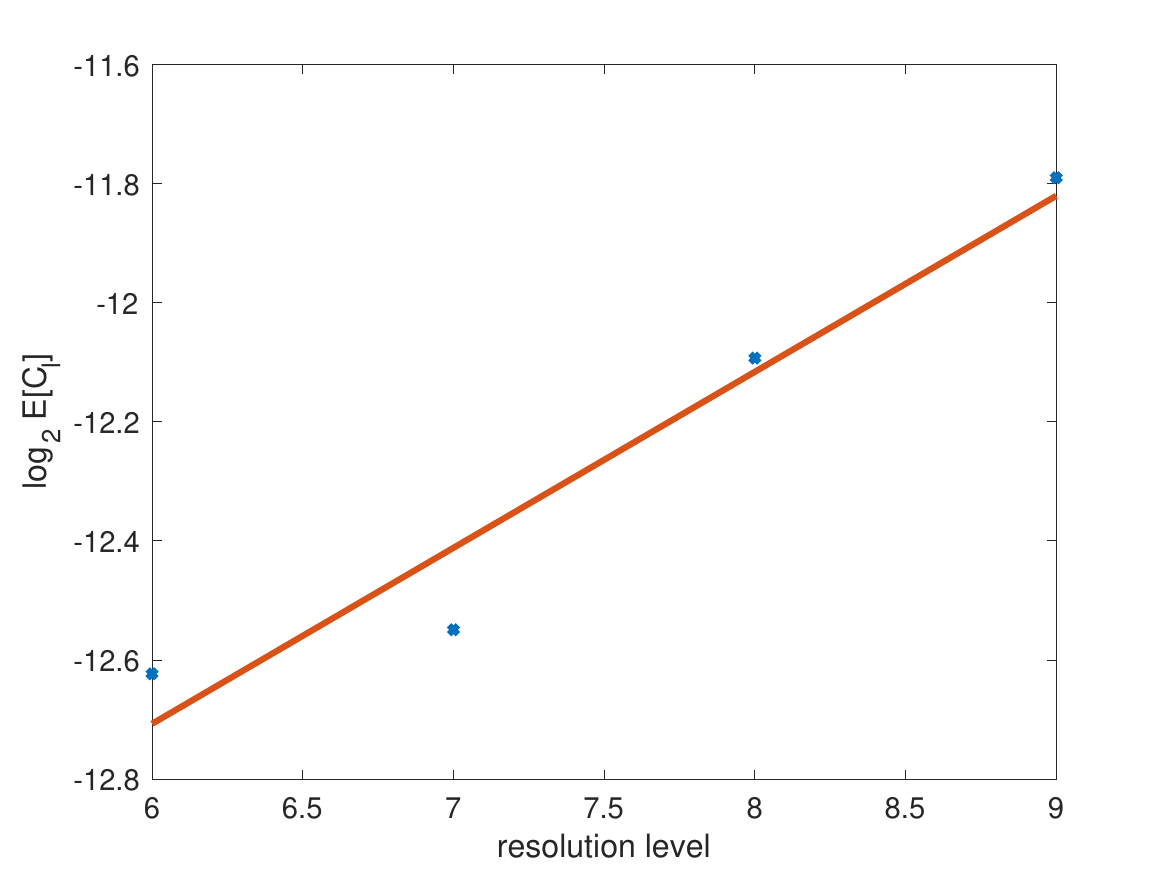}
  \caption{Expected cost on each resolution level. The slope of the plot is denoted as $\gamma$.}\label{fig:1p1C}
\endminipage
\caption{Bias, variance, and expected cost on each level for the Heston model. Blue dots represent estimates, red lines represent the line of best fit.}
\end{figure}

\begin{figure}[!htb]
\minipage{0.32\textwidth}
  \includegraphics[width=\linewidth]{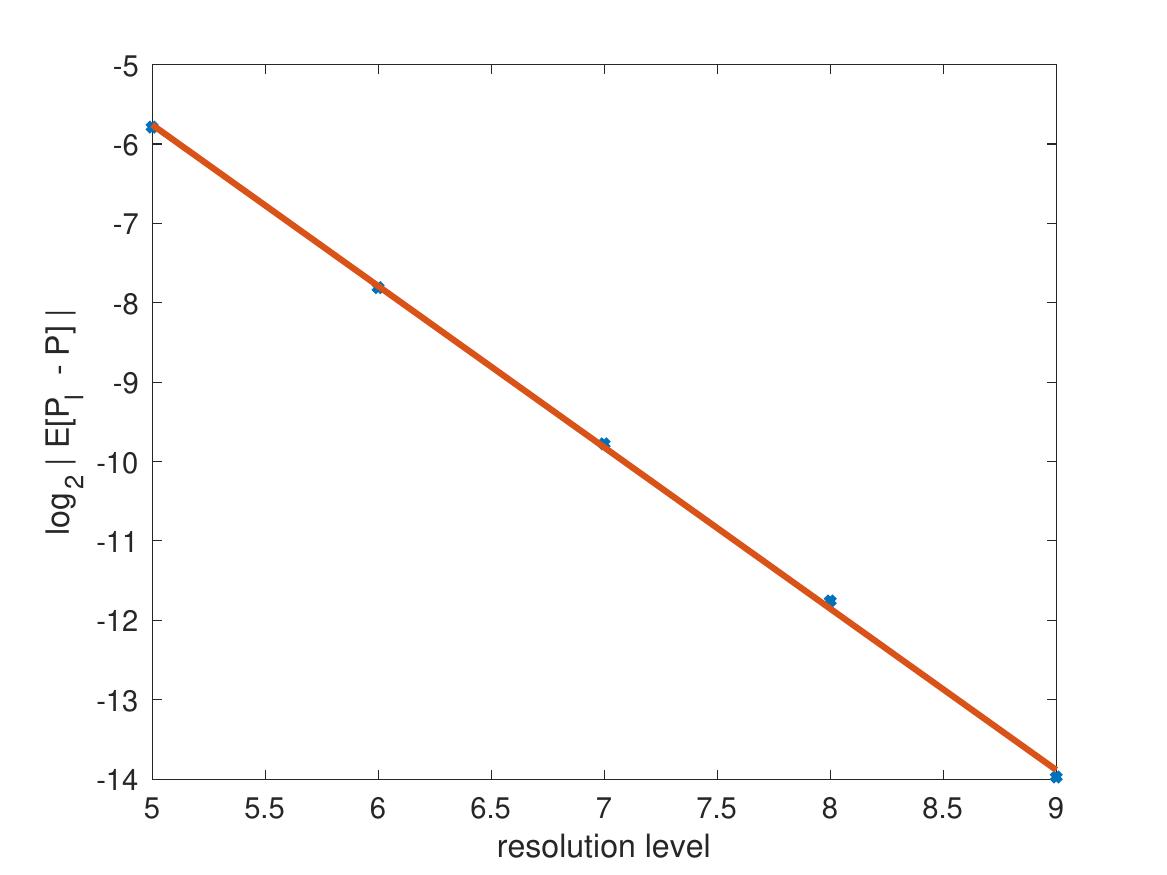}
  \caption{Bias on each resolution level. The slope of the plot is denoted as $\alpha$.}\label{fig:2p2B}
\endminipage\hfill
\minipage{0.32\textwidth}
  \includegraphics[width=\linewidth]{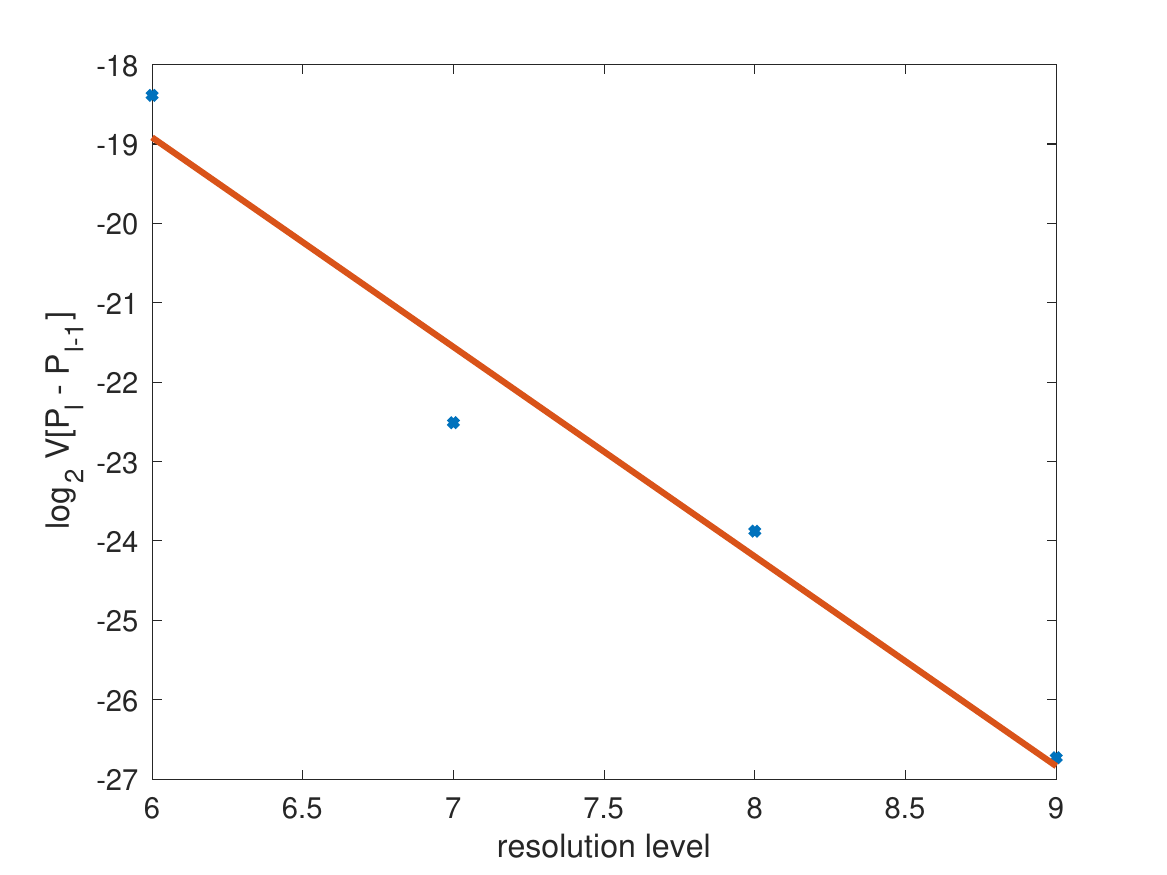}
  \caption{Variance on each resolution level. The slope of the plot is denoted as $\beta$}\label{fig:2p2V}
\endminipage\hfill
\minipage{0.32\textwidth}%
  \includegraphics[width=\linewidth]{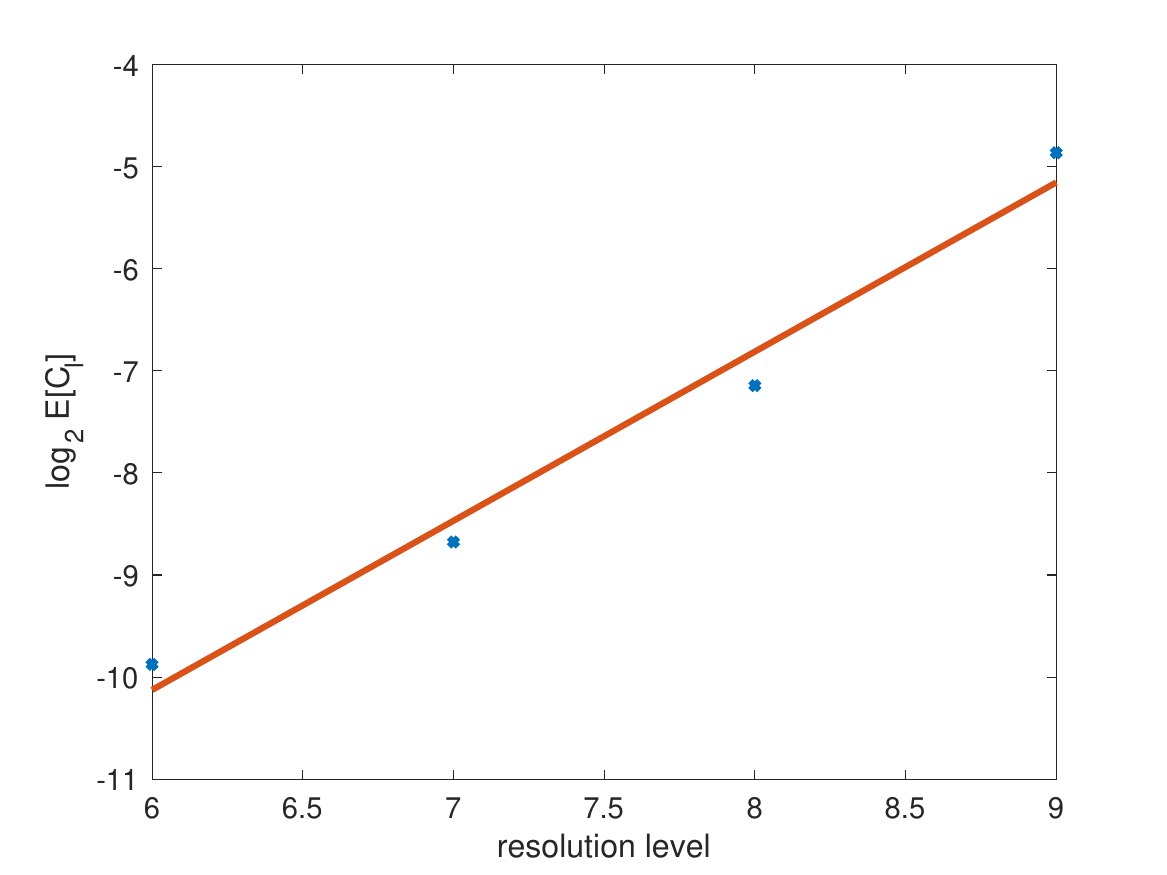}
  \caption{Expected cost on each resolution level. The slope of the plot is denoted as $\gamma$.}\endminipage
\caption{Bias, variance, and expected cost on each level for the Multi-Dimensional Heston model. Blue dots represent estimates, red lines represent the line of best fit.}
\label{fig:2p2C}
\end{figure}

\end{document}